\documentclass[12pt]{article}
\pdfoutput=1

\usepackage[cmex10]{amsmath}
\usepackage{amsthm}
\usepackage{amssymb}
\usepackage{tikz}
\usepackage{color}
\usepackage{xspace}
\usepackage{authblk}

\newcommand{\twod}{\ensuremath{\ell_2^2}\xspace}
\newcommand{\threed}{\ensuremath{\ell_2^3}\xspace}

\newtheorem{theorem}{Theorem}
\newtheorem{lemma}{Lemma}
\newtheorem{corollary}{Corollary}

\begin{document}

\title{Hilbert Exclusion: Improved Metric  Search \\through Finite Isometric Embeddings}

\author[1]{Richard Connor}
\affil[1]{Department of Computer and Information Sciences,University of Strathclyde, Glasgow, G1 1XH, United Kingdom}
\author[2]{Franco Alberto Cardillo}
\author[2]{Lucia Vadicamo}
\author[2]{Fausto Rabitti}
\affil[2]{ISTI (Information Science and Technology Institute)\\National Research Council of Italy,
Via Moruzzi 1, 56124 Pisa, Italy}
\affil[1]{{\footnotesize richard.connor@strath.ac.uk}}
\affil[2]{\footnotesize  \{franco.alberto.cardillo, lucia.vadicamo, fausto.rabitti\}@isti.cnr.it}

\date{}                     
\setcounter{Maxaffil}{0}
\renewcommand\Affilfont{\itshape\small}



%
%
\maketitle

\begin{abstract}
Most research into similarity search in metric spaces relies upon the triangle inequality property. This property allows the space to be arranged according to relative distances to avoid searching some subspaces. We show that many common metric spaces, notably including those using Euclidean and Jensen-Shannon distances, also have a stronger property, sometimes called the four-point property: in essence, these spaces allow an isometric embedding of any four points in three-dimensional Euclidean space, as well as any three points in two-dimensional Euclidean space. 
In fact, we show that 
any space which is isometrically embeddable in Hilbert space has the stronger property. This property gives stronger geometric guarantees, and one in particular, which we name the Hilbert Exclusion property, allows any indexing mechanism which uses hyperplane partitioning to perform better.  One outcome of this observation is that a number of state-of-the-art indexing mechanisms over high dimensional spaces can be easily extended to give a significant increase in performance; furthermore,  the improvement given is greater in higher dimensions. This therefore leads to a significant improvement in the cost of metric search in these spaces.
\end{abstract}

\section{Introduction}

In the realm of similarity search, many metric indexing techniques are available. These  rely on the  metric properties of the distance function used, and in particular use the triangular inequality property in various ways to exclude parts of the space from a search for  values similar to a given query.

Any proper metric space $(U, d)$ is  isometrically 3-embeddable in two dimensional Euclidean space (\twod). That is, for any three objects within $(U,d)$, there exists a function mapping those objects into \twod which preserves the distances between them. This is  in fact a  corollary of the  metric properties of $d$.

In this paper we consider spaces with the stronger property of being  isometrically 4-embeddable in three dimensional Euclidean space (\threed). We show that these spaces include all those which have isometric embeddings in Hilbert space, notably including any space under Euclidean distance, as well as the proper metric forms of Jensen-Shannon, Triangular Discrimination and a novel form of Cosine distance.

Such spaces give stronger  geometric properties. All metric indexing currently relies on one (or  both) of two core principles: exclusion based on a bounding radius, or exclusion based on a  hyperplane partition, both of which can be explained in terms of their 3-embeddabilty property. Using the stronger 4-embeddability, we show that a greater degree of exclusion is possible, and that this exclusion degrades more slowly as higher dimensions are considered.

Our main result is very simple. Consider any four points $p_1, p_2, q$ and $s$ in a metric space $(U,d)$, where the intent is that $q$ is a query, $s$ is a solution to this query (i.e. $d(q,s) \le t$ for some real value $t$), and $p_1$ and $p_2$ are points within the space which have been previously used to structure the data.

 During the progress of a query evaluation, the distances $d(q,p_1)$ and $d(q,p_2)$ are evaluated. Assuming without loss of generality that $d(q,p_2) \le d(q,p_1)$, then a well known property used during search is
 
 \[\frac{d(q,p_1) - d(q,p_2 )}{2} > t \Rightarrow  d(s,p_1) > d(s,p_2)\]
 
 Here, we show  that for certain common classes of spaces
 
 \[\frac{d(q,p_1)^2 - d(q,p_2)^2}{2\,d(p_1,p_2)} > t \Rightarrow  d(s,p_1) > d(s,p_2)\] 
 
 Both properties can be used to avoid searching subspaces  where all elements  are known to be closer to $p_1$ than $p_2$. The second property however is strictly weaker, meaning that any indexing mechanism which uses the first can be made more efficient\footnote{The distance $d(p_1,p_2)$ can be evaluated as the index is built, not during the query.}.

The best performing index for general purposes is currently believed to be the distal SAT \cite{dSatSisap,dSatIS}, which uses a combination of pivot and hyperplane-based exclusion. For this structure, we show a significant performance increase for Euclidean and Jensen-Shannon spaces, especially in higher dimensions. This therefore gives, for these spaces, a new high performance benchmark for similarity search.

The rest of this article is structured as follows. Section \ref{section_context} gives a general context of metric search and finite isometric embedding; after basic definitions, it goes on to show how the essential mechanisms of metric search can be explained in terms of  finite embeddings.  Section \ref{sec_3Dembedding_outline} briefly shows, in outline, why better performance can be expected from a space which is 4-embeddable in \threed. Section  \ref{section_better_exclusion}  gives a formal definition of our new exclusion property for hyperplane partitioning, and proves its applicability to any space which is isometrically 4-embeddable in \threed. Section \ref{section_four_embeddable_spaces}  gives some background mathematics of Hilbert spaces, and shows the 4-embeddabilty property for three important metrics.
Section \ref{sec_performance_analysis} gives an analysis of the improvement, including  relative performance measurements for some metric index implementations which use hyperplane partitioning.
 Section \ref{section_increasing_dimensionality} shows how the new exclusion criterion degrades relatively less severely over higher dimensions than those currently used, and   Section \ref{sec_conclusions} summarises and outlines further possibilities.

\section{Background and Related Work}
\label{section_context}

\subsection{Similarity Search and Metric Indexing}To set the context, we are interested in searching a (large) finite  set of objects $S$ which is a subset of an infinite set $U$, where $(U, d)$ is a  metric space. The general requirement is to efficiently find members of $S$ which are similar to an arbitrary member of $U$, where the distance function $d$ gives  the only way by which any two objects may be compared. There are many important practical examples captured by this mathematical framework, see for example \cite{Chavez05,zezula2006similarity}.

For $(U, d)$ to be a metric space, the distance function $d: U\times U\to \mathbb{R}$ requires  to satisfy 
 	\begin{itemize}
 		\item Positivity: $\forall\, a, b \in U,\, d(a,b)\geq 0$ with equality if, and only if, $a=b$;
 		\item Symmetry: $\forall\, a, b \in U,\, d(a,b)=d(b,a)$; 
 		\item Triangle inequality: $\forall\, a, b, c \in U,\, d(a,c)\leq d(a,b)+d(b,c)$.
 	\end{itemize}
 Such spaces are typically searched with reference to a query object $q	 \in U$.
 A threshold search  for some threshold $t$, based on a query $q \in U$, has the solution set  $\{s \in S \,\, \text{such that}\,\, d(q,s) \le t\}$. 
 	
 Typically $S$ is too large to allow an exhaustive search.
 However such queries can often be performed efficiently by use of a \emph{metric index}, one of a large family of data structures which make use of the triangle inequality property in order to arrange the set of objects $S$ in such a way as to minimise the time required to retrieve the query result.
%
Efficiency is primarily achieved by avoiding unnecessary distance calculations, although the efficient use of memory hierarchies is also extremely important. Both of these are optimised by structuring the set based on relative distances of objects from each other, so that triangle inequality can be used to determine subsets which do not need to be exhaustively checked. Such avoidance is normally referred to as \emph{exclusion}.

For exact metric search, almost all indexing methods can be divided into those which at each exclusion possibility use  a single ``pivot" point to give radius-based exclusion, and those which use two reference points to give hyperplane-based exclusion. Many variants of each have been proposed, including many hybrids; \cite{Chavez:2001}, \cite{zezula2006similarity} give excellent surveys.
In general the best choice seems to depend on the particular context of metric and data.

Here we focus particularly on  mechanisms which use hyperplane-based exclusion.
The simplest such index structure  is the Generalised Hyperplane Tree  \cite{GHT}. Others include  the Monotonous Bisector Tree \cite{Noltemeier1992}, the Metric Index \cite{MIndex2011}, and the Spatial Approximation Tree \cite{SAT2002}. This last has various  derivatives, notably including the Dynamic SAT \cite{DSAT} and the Distal SAT \cite{dSatIS}, which includes a variant $\textit{SAT}_\textit{out}$ which is  believed to be, at time of writing, the most efficient known general-use indexing structure for performing exact search \cite{dSatIS}; therefore an significant improvement on this, as we show here, is a significant result.

\subsection{Finite Isometric Embeddings}

\label{section_finite_embeddings}

An  isometric embedding of one metric space $(V,d_v)$ in another $(W,d_w)$ can be achieved when there exists a mapping function $f:V\rightarrow W$ such that $d_v(x,y) = d_w(f(x),f(y))$, for all $x,y \in V$. A finite isometric embedding occurs whenever this property is true for any finite selection of $n$ points from $V$, in which case the terminology used is that $V$ is isometrically $n$-embeddable in $W$.

The first observation to be made in this context is that any metric space is isometrically $3$-embeddable in \twod. This is  apparent from the triangle inequality property of a proper metric, as illustrated in Figure \ref{fig_property_equivalence}.
In fact the two properties are  equivalent: for any semi-metric space%
\footnote{a space where triangle inequality is not guaranteed} $(V,d_v)$  which is isometrically $3$-embeddable in  \twod, triangle inequality also holds.

Much work was done on finite isometric embeddings in the 1930s, but it does not appear to have been a ``hot topic" since then.
Blumenthal \cite{blumenthal1933note} provides an excellent and concise summary of this work as it pertains to ours. He attributes our observation above, that any semi-metric space which is 3-embeddable in \twod is a  metric space, to Menger. He uses the phrase \emph{the four-point property} to mean a semi-metric space which is isometrically 4-embeddable in \threed.  Wilson \cite{wilson1932relation} shows various properties of such spaces, and Blumethal points out that results given by Wilson, when combined with work by Menger in \cite{menger_1}, generalise to show that some spaces have the \textit{$n$-point property} (i.e. any $n$ points can be isometrically embedded in $\ell_2^{n-1}$.) This is in fact a more  general result than our  Lemma \ref{lemma_hyperplane} which uses a more modern formulation for high dimensional Euclidean space.

The most important results in finite isometric embeddings from our perspective are given by Schoenberg and Blumethal. \cite{Schoenberg}  shows an initially surprising result that if a kernel function $K$ has certain simple properties, then it can be used to construct a metric space which is isometrically embeddable in a Hilbert space. Blumenthal \cite{blumenthal1953} shows that any space which is isometrically embeddable in a Hilbert space has the $n$-point property for every possible integer $n$. In combination these are extraordinarily strong from our perspective: for any kernel function $K$ with the correct properties, we can construct a proper metric space with the four-point property. We expand  on this observation in Section  \ref{section_four_embeddable_spaces}.

Although normally expressed in terms of the property of triangle inequality, the properties of a metric space that allow indexing can be equally well expressed in terms of the geometric guarantees afforded according to the 3-embeddability property in \twod. To set the context, we briefly explain the two main indexing principles in terms of this property.
\begin{figure}[t]
	\centering
	\includegraphics[trim=10mm 0mm 10mm 10mm, width=0.98\columnwidth]{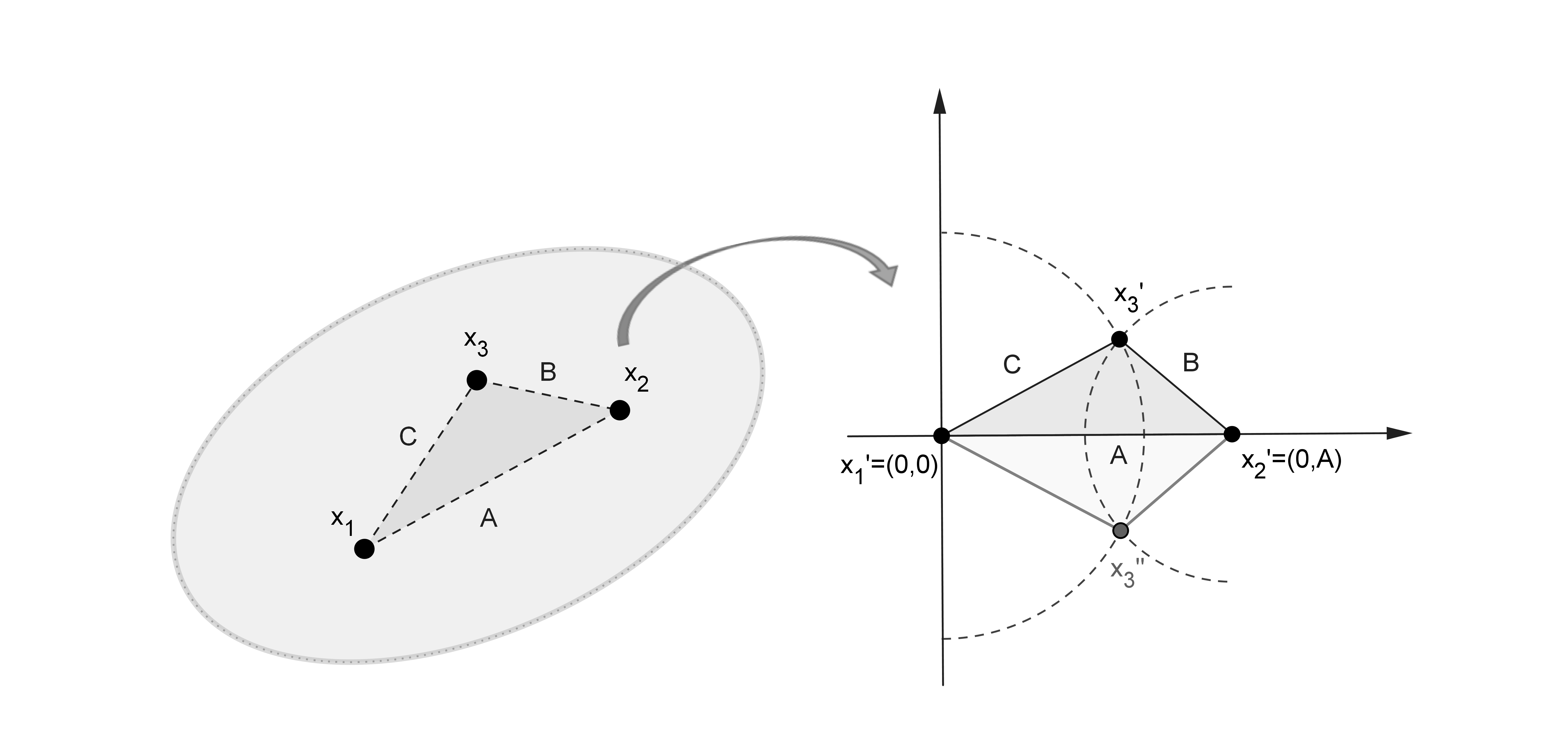}
	\caption{For any three points $x_1, x_2$ and $x_3$ whose distances satisfy the  triangle inequality property, a triangle can be constructed within 2D Euclidean space such that $x_1'$ is at the origin, $x_2'$ lies on the X-axis, and $x_3'$ is where the distances $B$ and $C$ intersect.}
\label{fig_property_equivalence}
\end{figure}

\subsection{Pivot-based indexing}
This technique entails the selection of a \emph{pivot} point $p \in S$, and the construction of one or more subsets of $S$ based on a fixed distance $m$  from $p$, e.g. $S_{in}$ where $s \in S_{in} \Rightarrow d(p,s) \le m$. For  a query $q$, $d(q,p)$ is calculated; if this is greater than $m + t$, for a query threshold $t$, then no element of $S$ within distance $t$ of $q$ can be within $S_{in}$ and every element of $S_{in}$ can therefore  be excluded from the search. Similarly, $S_{out}$ could be constructed such that $s \in S_{out} \Rightarrow d(p,s) > m$, in which case the elements of $S_{out}$ can be excluded if $d(q,p) \le m - t$.

The validity of the pivoting principle can be shown algebraically using the triangle inequality property of the metric, and many different mechanisms have been described using it \cite{Chavez:2001,zezula2006similarity}. They are often illustrated in the manner of Figure  \ref{figure_pivot_principle}; using such illustrations relies upon  isometric 3-embeddability   within \twod {} of any metric space, but should also be treated with care whenever more than three objects are considered, as  consideration of more than three points within the plane is invalid.

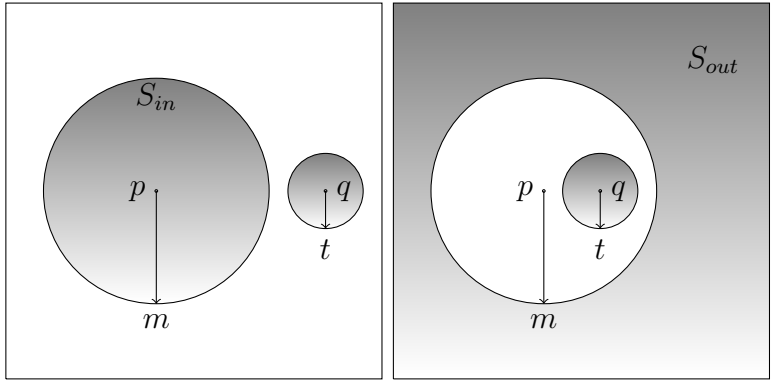
\begin{figure}[tbp]
\begin{center}
\begin{tikzpicture}[scale=0.5]

\draw (-5,-5) -- (5,-5) -- (5,5) -- (-5,5) -- (-5,-5);

\shade (-1,0) circle (3);
\draw (-1,0) circle (3);
\draw (-1,2.5) node {$S_{in}$};
\draw (-1,0) circle (1pt) node[anchor=east] {$p$};
\draw [->](-1,0) -- (-1,-3) node[anchor=north] {$m$};

\shade (3.5,0) circle (1);
\draw(3.5,0)circle(1);
\draw (3.5,0) circle (1pt) node[anchor=west] {$q$};
\draw [->](3.5,0) -- (3.5,-1) node[anchor=north] {$t$};

%
%
\end{tikzpicture}
\begin{tikzpicture}[scale=0.5]

\shade (-5,-5) -- (5,-5) -- (5,5) -- (-5,5) -- (-5,-5);
\draw (-5,-5) -- (5,-5) -- (5,5) -- (-5,5) -- (-5,-5);
\draw (3.5,3.5) node {$S_{out}$};

\fill[color=white] (-1,0) circle (3);
\draw (-1,0) circle (3);
\draw (-1,0) circle (1pt) node[anchor=east] {$p$};
\draw [->](-1,0) -- (-1,-3) node[anchor=north] {$m$};

\shade[color=white]  (0.5,0) circle (1);
\draw(0.5,0)circle(1);
\draw (0.5,0) circle (1pt) node[anchor=west] {$q$};
\draw [->](0.5,0) -- (0.5,-1) node[anchor=north] {$t$};

%
%
\end{tikzpicture}
\caption{Pivot-based exclusion illustrated by 3-embedding in \twod. Objects in $S_{in}$ are at most distance  $m$ from $p$, and objects in $S_{out}$ are at least $m$ from $p$. Given  $d(q,p) > m + t$, $S_{in}$ cannot contain a solution to the query. Similarily if  $d(q,p) < m - t$, $S_{out}$ cannot. Such diagrams should be treated with extreme care: for a general metric space, no more than three objects have a guarantee of isometric embedding within 2D Cartesian space. In these cases, it is necessary only to consider an embedding of the pivot, the query, and an arbitrary object within the solution space to see that the distance guarantee holds.}
\label{figure_pivot_principle}
\end{center}
\end{figure}



\subsection{Partition-based indexing}

In this type of indexing, two elements of $S$ are chosen, and the rest of $S$ is divided into two subsets according to which of these elements is closer.
Formally:
\begin{align*}
p_1, p_2 & \in S	\\
S_{p_1} & = \{s \in S - \{p_1, p_2\} \,, \, d(s,p_1) < d(s,p_2)\}	\\
S_{p_2} & = \{s \in S - \{p_1, p_2\} \,, \,d(s,p_1) \ge d(s,p_2)\}
\end{align*}

To evaluate a  query over  $q$, the distances $d(q,p_1)$ and $d(q,p_2)$ are first calculated. If $|d(q,p_1) - d(q,p_2)| > 2t$, then the subset associated with the point further from $q$ does not intersect with the solution set of the query and these values can be excluded from the search. Again, the exclusion condition is straightforward to derive algebraically from the triangle inequality property,  but can also be shown in terms of 3-embeddability within \twod.
%

\begin{figure}[tbp]
\begin{center}
\includegraphics[width=0.8\columnwidth]{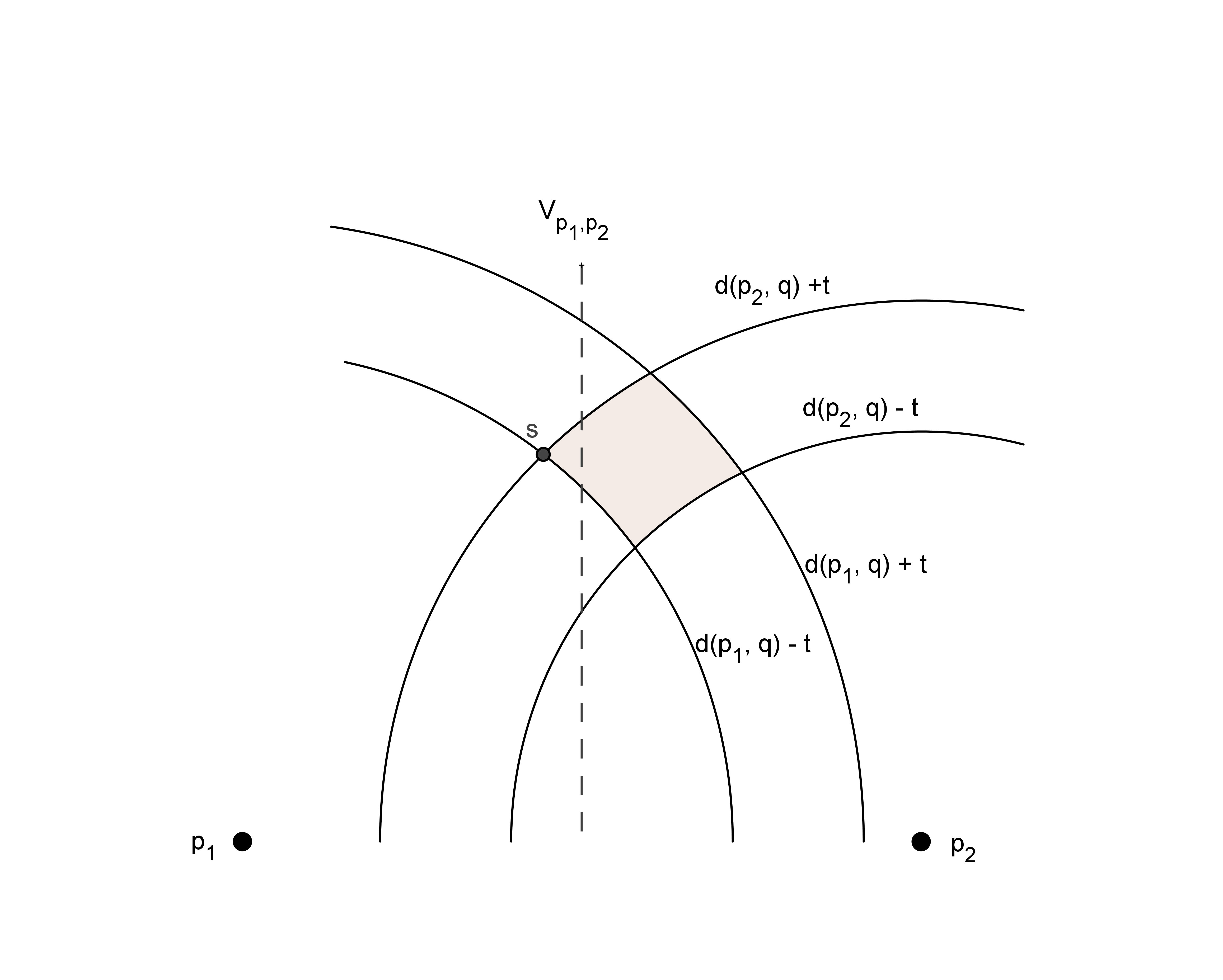}
%
%
%
%
%
%
%
%
%
%
%
\caption{The two pivot points and any solution to the query can  be isometrically  embedded in \twod.  The point $q$   cannot  be drawn in the same diagram. Given its distance from $p_1$ and $p_2$,  any solution in the original metric space must lie in the region bounded by  the four arcs shown in the \!{ \twod} projection. If the  point $s$ lies to the right of $V_{p_1,p_2}$ in \twod,  there is therefore no requirement to search to the left of the hyperplane in the original space. By symmetry, if $|d(q,p_1) - d(q,p_2)|>2t$, then half of the search space can be excluded. }
\label{figure_hyperplane_principle}
\end{center}
\end{figure}
%
%
%
%
%
%
%
%
Figure \ref{figure_hyperplane_principle} shows a graphical interpretation of this situation using the \twod {} embedding.  The three points chosen for illustration, relying on the 3-embedding property, are $p_1$, $p_2$, and an arbitrary solution  point to the query $q$. 
The two pivot points and any solution to the query can  be isometrically  embedded in \twod. In general the point $q$ may  not be, and  therefore cannot be drawn in the diagram

The line $V_{p_1,p_2}$  represents a boundary between $S_{p_1}$ and $S_{p_2}$  in the original space. If the whole of the region bounded by the four arcs lies to one side of this line, there is no requirement to search in the other part of the space. It can be seen from the diagram, if $q$ is closest to $p_2$,  that this occurs when $d(q,p_1) - t > d(q,p_2) + t$, i.e. $d(q,p_1) - d(q,p_2)  > 2t$. 
This illustration alone in fact is not quite convincing; it must be further observed that, for any two 3-embeddings where two of the points are the same (in this case $p_1$ and $p_2$), then embedding functions can be chosen that map those two points to the same two points in \twod (e.g. see Figure \ref{fig_property_equivalence}) thus preserving the semantics of the line $V_{p_1,p_2}$.

\section{Partition-based indexing with 4-embedding in \threed}
\label{sec_3Dembedding_outline}
We introduce the main result of this paper with simple   observation that, for spaces that are isometrically 4-embeddable in \threed, a tighter exclusion condition is possible for partitions.

Figure \ref{figure_queries} shows an example taken from a metric space 3-embedded in \twod, that is a standard metric space. Of the three queries, only $q_1$ and $q_2$ allow the partition on the far side of the hyperplane to be excluded, as for $q_3$ the exclusion condition is not met, even although the solution space appears geometrically separated from the right-hand side. 

This is because the  boundary defined by the  exclusion condition  is given by the locus of points $x$ such that { $d(x,p_2) - d(x,p_1) = 2t$} which  defines a  hyperbola focussed at $p_1$ and $p_2$, with semi-major axis $t$.
The minimum distance of this hyperbola from the  line $V_{p_1,p_2}$ is  $t$, but this occurs only on the line passing through $p_1$ and $p_2$. When considering this diagram in two dimensions, the relative distances among $p_1, p_2$ and any individual $q_i$ are significant, but as a general metric space  guarantees only 3-embeddability, the circles drawn around the queries are meaningless with respect to the original space.

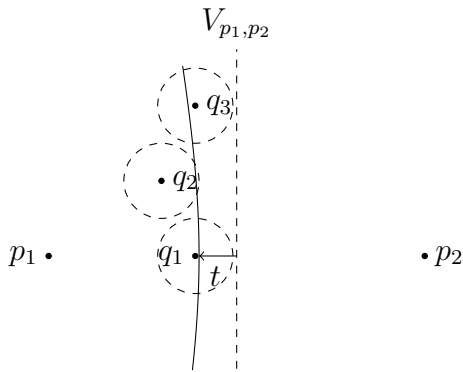
\begin{figure}[tbp]
	\begin{center}
		
		\begin{tikzpicture}[scale=0.5]
		
		\filldraw (-5,0) circle (2pt) node[anchor=east] {$p_1$};
		\filldraw (5,0) circle (2pt) node[anchor=west] {$p_2$};
		\filldraw (-1.1, 0) circle (2pt) node[anchor= east] {$q_1$};
		\filldraw (-2,2) circle (2pt) node[anchor=west] {$q_2$};
		\filldraw (-1.1,4) circle (2pt) node[anchor=west] {$q_3$};
		
		\draw[dashed]  (-0.1,4)arc(0:360:1.0);
		\draw[dashed] (-1,2)arc(0:360:1.0);
		\draw[dashed] (-0.1,0)arc(0:360:1.0);

		\draw[dashed] (0,-3) -- (0,5.5) node[anchor= south] {$V_{p_1,p_2}$};
		
		\draw[->] (0,0)-- (-1,0)  node[anchor=north west] {$t$};
		
		\draw
		(-1,0) --
		(-1.00006797710266,0.0571227785479771) --
		(-1.00027196387189,0.114263034245484) --
		(-1.00061212680032,0.171438266430466) --
		(-1.00108874373915,0.228666018878084) --
		(-1.00170220444469,0.285963902170542) --
		(-1.00245301134638,0.343349616249187) --
		(-1.00334178053901,0.400840973210953) --
		(-1.00436924300282,0.458455920412389) --
		(-1.00553624605596,0.516212563945954) --
		(-1.00684375504477,0.574129192555036) --
		(-1.00829285527797,0.632224302056251) --
		(-1.00988475421220,0.690516620339992) --
		(-1.01162078389706,0.749025133023020) --
		(-1.01350240368889,0.807769109830058) --
		(-1.01553120324369,0.866768131784914) --
		(-1.01770890580079,0.926042119295686) --
		(-1.02003737176986,0.985611361223092) --
		(-1.02251860263561,1.04549654502590) --
		(-1.02515474519540,1.10571878808296) --
		(-1.02794809614705,1.16629967029751) --
		(-1.03090110704520,1.22726126809589) --
		(-1.03401638964688,1.28862618994062) --
		(-1.03729672166831,1.35041761348549) --
		(-1.04074505297735,1.41265932450956) --
		(-1.04436451224807,1.47537575777682) --
		(-1.04815841410634,1.53859203997880) --
		(-1.05213026679792,1.60233403492978) --
		(-1.05628378041334,1.66662839119680) --
		(-1.06062287570706,1.73150259236157) --
		(-1.06515169355160,1.79698501012670) --
		(-1.06987460507113,1.86310496049646) --
		(-1.07479622250307,1.92989276328106) --
		(-1.07992141084048,1.99737980519499) --
		(-1.08525530031321,2.06559860684252) --
		(-1.09080329977064,2.13458289390981) --
		(-1.09657111103533,2.20436767291081) --
		(-1.10256474430270,2.27498931186559) --
		(-1.10879053466955,2.34648562632435) --
		(-1.11525515988180,2.41889597118840) --
		(-1.12196565940069,2.49226133882190) --
		(-1.12892945489628,2.56662446399506) --
		(-1.13615437228783,2.64202993625172) --
		(-1.14364866546261,2.71852432035209) --
		(-1.15142104181772,2.79615628550599) --
		(-1.15948068978455,2.87497674418393) --
		(-1.16783730851157,2.95503900137331) --
		(-1.17650113989978,3.03639891523681) --
		(-1.18548300320545,3.11911507023009) --
		(-1.19479433244796,3.20324896384820) --
		(-1.20444721688588,3.28886520829594) --
		(-1.21445444485385,3.37603174851872) --
		(-1.22482955128492,3.46482009818974) --
		(-1.23558686927985,3.55530559542812) --
		(-1.24674158612600,3.64756768022521) --
		(-1.25830980421540,3.74169019578460) --
		(-1.27030860736411,3.83776171624039) --
		(-1.28275613309524,3.93587590351153) --
		(-1.29567165151588,4.03613189638366) --
		(-1.30907565149591,4.13863473528970) --
		(-1.32298993494471,4.24349582669289) --
		(-1.33743772008333,4.35083345147120) --
		(-1.35244375472478,4.46077332226798) --
		(-1.36803444070809,4.57344919542393) --
		(-1.38423797078422,4.68900354385230) --
		(-1.40108447942774,4.80758829808036) --
		(-1.41860620925103,4.92936566367460) --
		(-1.43683769493288,5.05450902441780) ;
		
		\draw
		(-1,-0) --
		(-1.00006797710266,-0.0571227785479771) --
		(-1.00027196387189,-0.114263034245484) --
		(-1.00061212680032,-0.171438266430466) --
		(-1.00108874373915,-0.228666018878084) --
		(-1.00170220444469,-0.285963902170542) --
		(-1.00245301134638,-0.343349616249187) --
		(-1.00334178053901,-0.400840973210953) --
		(-1.00436924300282,-0.458455920412389) --
		(-1.00553624605596,-0.516212563945954) --
		(-1.00684375504477,-0.574129192555036) --
		(-1.00829285527797,-0.632224302056251) --
		(-1.00988475421220,-0.690516620339992) --
		(-1.01162078389706,-0.749025133023020) --
		(-1.01350240368889,-0.807769109830058) --
		(-1.01553120324369,-0.866768131784914) --
		(-1.01770890580079,-0.926042119295686) --
		(-1.02003737176986,-0.985611361223092) --
		(-1.02251860263561,-1.04549654502590) --
		(-1.02515474519540,-1.10571878808296) --
		(-1.02794809614705,-1.16629967029751) --
		(-1.03090110704520,-1.22726126809589) --
		(-1.03401638964688,-1.28862618994062) --
		(-1.03729672166831,-1.35041761348549) --
		(-1.04074505297735,-1.41265932450956) --
		(-1.04436451224807,-1.47537575777682) --
		(-1.04815841410634,-1.53859203997880) --
		(-1.05213026679792,-1.60233403492978) --
		(-1.05628378041334,-1.66662839119680) --
		(-1.06062287570706,-1.73150259236157) --
		(-1.06515169355160,-1.79698501012670) --
		(-1.06987460507113,-1.86310496049646) --
		(-1.07479622250307,-1.92989276328106) --
		(-1.07992141084048,-1.99737980519499) --
		(-1.08525530031321,-2.06559860684252) --
		(-1.09080329977064,-2.13458289390981) --
		(-1.09657111103533,-2.20436767291081) --
		(-1.10256474430270,-2.27498931186559) --
		(-1.10879053466955,-2.34648562632435) --
		(-1.11525515988180,-2.41889597118840) --
		(-1.12196565940069,-2.49226133882190) --
		(-1.12892945489628,-2.56662446399506) --
		(-1.13615437228783,-2.64202993625172) --
		(-1.14364866546261,-2.71852432035209) --
		(-1.15142104181772,-2.79615628550599) --
		(-1.15948068978455,-2.87497674418393) --
		(-1.16783730851157,-2.95503900137331) --
		(-1.17650113989978,-3.03639891523681);
		
		\end{tikzpicture}
		
		\caption{Three queries, $q_1, q_2$ and $q_3$, each with threshold $t$, on the left side of the  boundary $V_{p_1,p_2}$. Since $d(p_2,q_3) - d( p_1,q_3) < 2t$, $S_{p_2}$ cannot be excluded from the search on $q_3$. ($p_1 = (-5,0), p_2=(5,0), q_3 = (-1.1,4), t=1$). The hyperbola curve represents all possible points $x\in S_{p_1}$ such that $d(p_2,x) - d( p_1,x)= 2t$, i.e. the boundary of the exclusion condition.}
		\label{figure_queries}
	\end{center}
\end{figure}
\begin{figure}[tbp]
	\begin{center}	
{\includegraphics[width=0.76\columnwidth]{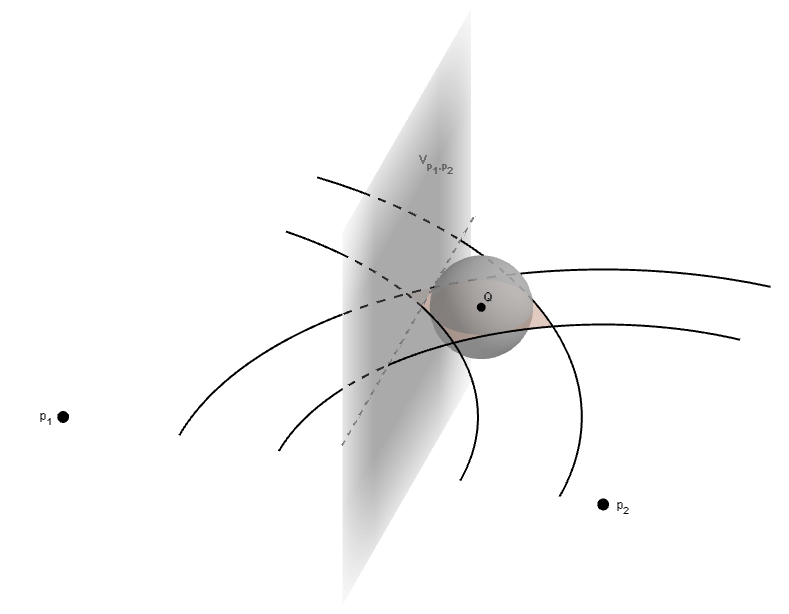}}
%
%
%
%
%
%
%
		\caption{Four points ($p_1,p_2,q$ and $s,  \text{  s.t.  } d(q,s) \le t$) in \threed. For fixed $p_1, p_2$ and $q$, any solution to the query lies within the sphere centred around $q$ and cannot lie within $S_{p_1}$, even although $d(q,p_1) - d(q,p_2) < 2t$. Note that $V_{p_1,p_2}$ in the figure now represents the hyperplane that divides the space into two subspaces: objects nearer to $p_1$ belonging to the left subspace and objects nearer to $p_2$ to the right.
}
		\label{figure_4_embeddable}
	\end{center}
\end{figure}

Consider now Figure \ref{figure_4_embeddable}, which shows the same situation  but relying on a 4-embeddability in \threed. Here 
the relative distances among any four points can be safely considered: in this case $p_1, p_2, q$, and any solution to $q$. The plane on which the diagram is drawn is that containing $p_1, p_2$ and $q$, and therefore the locus of any solution to $q$ consists of a sphere, radius $t$, centred around $q$.

It is clear from this diagram, in comparison with  Figure \ref{figure_hyperplane_principle}, that a more useful exclusion condition can be used: whenever the distance between $q$ and  $V_{p_1,p_2}$ is greater than $t$, $S_{p_1}$ does not require to be searched.  Other than the single point on the line through  $p_1$ and $p_2$ this distance is always strictly less than the nearest point on the corresponding hyperbola, and thus more exclusions are always possible.

Figure \ref{fig_3D_illustration} gives an illustration of the two boundary conditions in \threed.
It can be seen that our new exclusion condition is   weaker than the normal, hyperbolic, condition; in this sense weaker implies better, as it allows more queries to exclude the opposing semispace from further consideration. For discussion in the rest of the paper, we refer to the new exclusion condition as \emph{Hilbert Exclusion}, and the former condition as \emph{Hyperbolic Exclusion}. We proceed with a formal definition and proof of correctness of Hilbert Exclusion.

\begin{figure}[tbp]
\centering
\includegraphics[width=0.7\columnwidth]{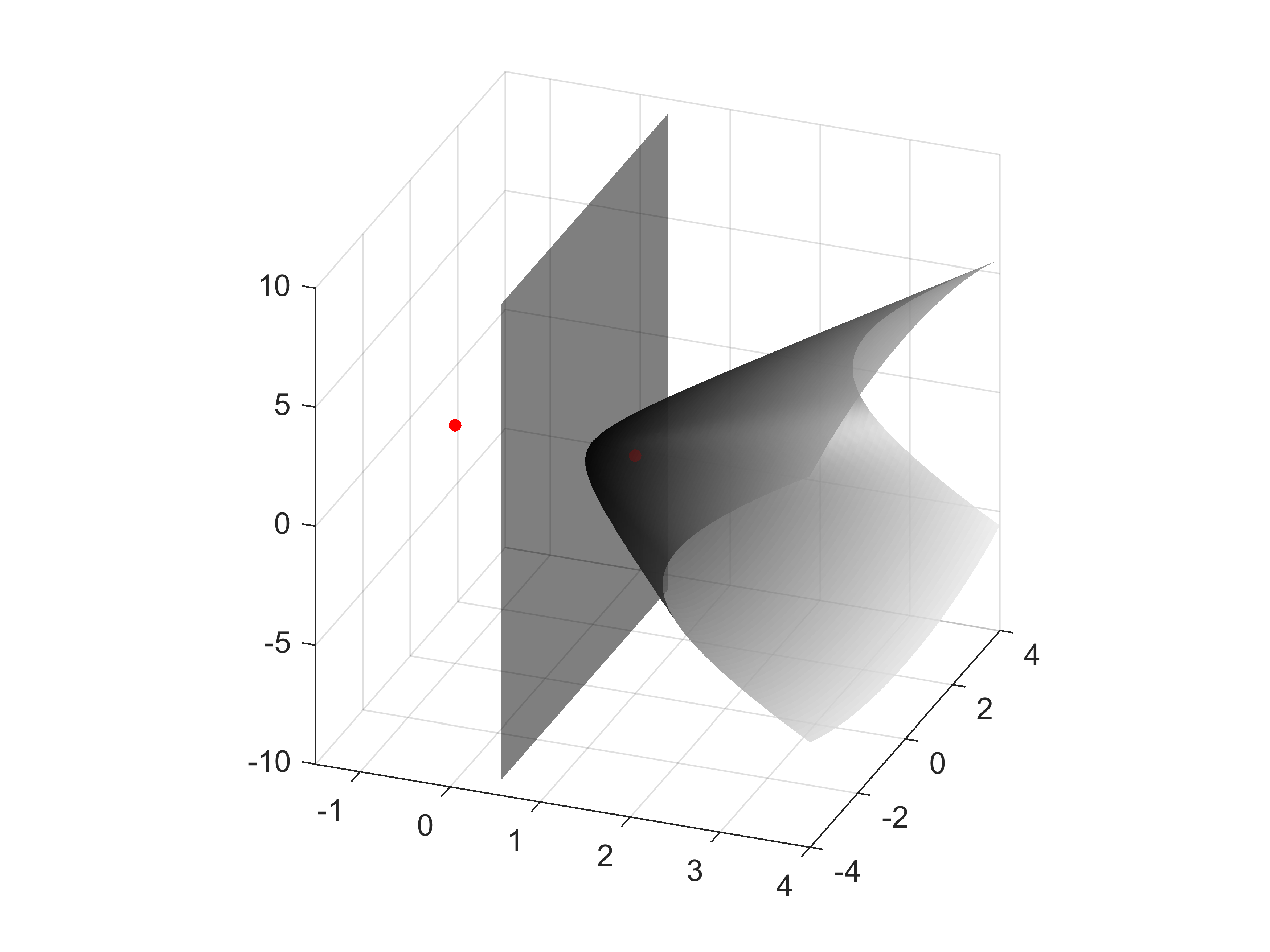}
\caption{This illustration shows the geometric principle behind the new exclusion condition, which can be applied to any metric space which is isometrically 4-embeddable in 3D Euclidean space. Here pivots are placed at $(-1,0,0)$ and $(1,0,0)$, the threshold selected is $0.5$. The surfaces drawn represent the boundaries of the two exclusion conditions we now refer to as \emph{Hilbert Exclusion} and \emph{Hyperbolic Exclusion} }
\label{fig_3D_illustration}
\end{figure}
\section{The Hilbert Exclusion Condition}
\label{section_better_exclusion}

\begin{theorem}
	\label{theorem_2d}
	Consider any three points $p_1, p_2, q \in \,$\threed with $d(q,p_2) < d(q,p_1)$. Then the condition
	\begin{equation}
		\dfrac{d(q,p_1)^2 - d(q ,p_2)^2}{2\,d(p_1,p_2)} > t
	\end{equation}
	
	implies that $d(s,p_2) < d(s,p_1)$ for all $s$ s.t. $d(q,s) \le t$.
	\end{theorem}
\begin{proof}

	It is sufficient to prove that the distance between the point $q$ and the plane  $V_{p_1,p_2}$
	is greater that $t$. In this case, $d(s,p_2) < d(s,p_1)$ for all $s$ s.t. $d(q,s) \le t$.
		
		 The equation of the plane  $V_{p_1,p_2}$ can be written as the scalar product $(p_2-p_1)\cdot (x-\frac{(p_2+p_1)}{2})=0$, and so its distance from $q$ is given by
		\[
		 \text{dist}(q,V_{p_1,p_2}) = \left	| \left(q- \frac{(p_2+p_1)}{2}\right)\cdot \frac{(p_2-p_1)}{\|p_2-p_1\|_2}\right|
		=\dfrac{d(q,p_1)^2 - d(q ,p_2)^2}{2\,d(p_1,p_2)}
		\]
		Therefore if  $\text{dist}(q,V_{p_1,p_2}) > t$, any  point within distance $t$ of $q$ is closer to $p_2$ than to $p_1$
%
%
%
%
%
%
%
\end{proof}

The practical application of this theorem is in search indexes which partition the search space. The exclusion condition
\[\frac{d(q,p_1)^2 - d(q,p_2)^2}{2\,d(p_1,p_2)} > t\]
can be used in place of
\[\frac{d(q,p_1) - d(q,p_2 )}{2} > t\]
in order to exclude any subspace which is known to be closer to $p_1$ than to $p_2$.
The important point in our context is that the first condition is weaker than the second%
\footnote{A simple proof is given in Appendix \ref{section_appendix_proof_1}.}%
, and therefore will always result in more exclusions being made.

		
%
%
%
%
%
%
%
%
%
%
%
%
%
%

\begin{theorem}
	For any metric space $(U, d)$, and for any three points $p_1, p_2, q \in U$, the exclusion condition of Theorem \ref{theorem_2d} holds if $(U, d)$ is isometrically $4$-embeddable in \threed.
\end{theorem}

\begin{proof}
	Let $(U,d)$ be a metric space isometrically 4-embeddable in \threed.  
	Let $t$ be a real positive number and  $p_1, p_2,q \in U$ be three points such that   $d(q,p_2) < d(q,p_1)$ and 
	\begin{equation}\label{thcond}
		\dfrac{d(q,p_1)^2 - d(q ,p_2)^2}{2\,d(p_1,p_2)} > t .
	\end{equation}

	
	For any $s\in U$ such that $d(q,s)\leq t$ we want to prove that $d(s,p_2)<d(s,p_1)$.
	Since $(U,d)$ is isometrically 4-embeddable in \threed, there exists a function 
	$f:(U,d)\to\ell^3_2$
	which preserves all the six distances: 
	\begin{align}
		&	\|f(p_1)-f(p_2)\|_2=d(p_1,p_2)	\label{eq1} \\ 
		&	\|f(q)-f(p_1)\|_2=d(q,p_1)  \label{eq2}\\
		& 	\|f(q)-f(p_2)\|_2=d(q,p_2) \label{eq3}\\
		& 	\|f(s)-f(q)\|_2=d(s,q)\leq t \label{eq4}\\
		&\|f(s)-f(p_1)\|_2=d(s,p_1)\\
		&\|f(s)-f(p_2)\|_2=d(s,p_2).
	\end{align}
	Equations \eqref{eq1}-\eqref{eq4} together with equation \eqref{thcond} imply that points $\{f(p_1)$, $f(p_2)$, $f(q), f(s)\} \in \,$ \!{\threed} satisfy the exclusion condition of Theorem 1.
	Thus,  $f(s)$ is closer to $f(p_2)$ than to $f(p_1)$, i.e., $\|f(s)-f(p_1)\|_2>\|f(s)-f(p_2)\|_2$. This proves also that $s$ is closer to $p_2$ than to $p_1$, in fact
	\begin{equation*}
		d(s,p_1)=\|f(s)-f(p_1)\|_2>\|f(s)-f(p_2)\|_2=d(s,p_2).
	\end{equation*}
\end{proof}

Note that, for any solution $s$ in $U$, a different mapping function $f$ may be required, however the only importance of this function is that, for any four points, it exists: there is no requirement to identify it.

\section{Vector Spaces Isometrically  4-Embeddable in \threed}
\label{section_four_embeddable_spaces}
\subsection{$\ell_2^n$ Space}

Euclidean distance applied over many-dimensional data is probably the most common of metric searches. In these cases, we have an immediate result:

\begin{theorem}
Any $n$-dimensional Euclidean space  (i.e. an $\ell_2^n$ space, for any $n$) is 4-embeddable in \threed
\end{theorem}
\begin{lemma}
\label{lemma_hyperplane}
In n dimensions, precisely one k-dimensional hyperplane passes through any $(k + 1)$ points that do not lie in a $(k - 1)$-dimensional hyperplane.%
\footnote{If the points are coplanar, an infinity of such hyperplanes exist; the important point for our purposes is only that at least one such hyperplane exists.} 
Moreover,  a k-dimensional hyperplane  can be regarded as a k-dimensional space in its own right. (See for example \cite{kolmogorov_textbook}, Chapter 7.)
\end{lemma}
\begin{proof}
From Lemma  \ref{lemma_hyperplane}, any $\ell_2^n$ space is $(k+1)$-embeddable in $\ell_2^{k}$.
Therefore any $\ell_2^n$ space is $4$-embeddable in \threed.
\end{proof}

\begin{corollary}
The Hilbert Exclusion Condition is valid over Euclidean spaces of any dimension.
\end{corollary}

However, we have a  more general result: any metric space which has an isometric embedding in a Hilbert space is also $4$-embeddable in \threed. This  includes Euclidean space of any dimension, but also  includes other important spaces,  notably any governed by the Jensen-Shannon distance.

\subsection{Inner Product Spaces and Hilbert Spaces}


The importance of Hilbert spaces is the generalisation of the notion of Euclidean space by extending the methods of vector algebra and calculus to spaces with any finite or infinite number of dimensions. A Hilbert space is an abstract vector space possessing the structure of an inner product that allows length  and  angle to be measured which gives certain geometric properties. These properties extend to  abstract, non-geometric spaces which can be isometrically embedded in a Hilbert space.  The key property of interest here is in 4-point isometric embedding in \threed.

\begin{lemma}[Shoenberg's Theorem \cite{Schoenberg,topsoe2003jenson}]
\label{lemma_schoenberg}
	Let $X$ be a nonempty set and $K: X\times X \to \mathbb{R}$ a mapping that satisfies the positivity and symmetric proprieties and such that, for all finite sets $(c_i)_{i\leq n}$ of real numbers and all finite sets $(x_i)_{i\leq n}$ of points in $X$, the implication
	\begin{equation}
	\label{eqn_neg_semi_def}
	\sum_{i=1}^{n} c_i=0\Rightarrow \sum_{i,j=1}^{n} c_i c_j K(x_i, x_j)\leq 0
	\end{equation}
	holds (i.e., K is conditionally negative semidefinite function). Then $(X, \sqrt{K}) $ is a metric space which can be embedded isometrically as a subspace of a real Hilbert space. 
%
%
\end{lemma}
%
%
%

The main importance from our perspective is that, given a metric space $(X,\sqrt{K} )$, it is sufficient for $K$ to be a conditionally negative semidefinite function in order to have isometric embeddability into a Hilbert Space.

\begin{lemma}[Blumenthal Lemma 53.1 \cite{blumenthal1953}]  
A numerable semimetric space is isometrically embeddable in a Hilbert space if and only if it is isometrically $n$-embeddable in $\ell^{n-1}_2$ for every positive integer $n$.   
\end{lemma}

\begin{lemma}[Scholtes Proposition 1.3 \cite{scholtes2013characterisation}]
		Let $(X,\|\cdot\|)$ be a normed vector space. Then the following statements are equivalent:
		\begin{itemize}
			\item $(X,\|\cdot\|)$ is an inner product space, i.e., there exists an inner product $<\cdot,\cdot>$ on $X$ which induces the norm: $\forall\, x\in X, \, \|x\|=\sqrt{<x,x>}$ 
			\item  all subsets $\{u,v,w,x\}\subset X$ are isometrically embeddable in \threed.
		\end{itemize}
\label{lemma_scholtes}
\end{lemma}

By definition, any Hilbert space is a normed vector space which is also an inner product space. From the above lemmata, we can observe that for any semimetric, negative semidefinite kernel function $K$ over $\mathbb{R}^n$, then $(\mathbb{R}^n,\sqrt{K})$ is a proper metric space which can be searched using  our new exclusion rule. The fact that the resulting metric space is a subspace of  Hilbert space is not strictly necessary for this purpose, although it gives other potentially valuable geometric properties as well. 
In fact, the Hilbert embeddability guarantees the $n$-point property for all  $n$, while just the $4$-point property is required for our new exclusion rule.
It is worth noting that in \cite{blumenthal1953} a weaker version of the Schoenberg's theorem is used to characterise any  metric space which has the 4-point property: 

\begin{lemma}[\cite{blumenthal1953}]
A metric space $(X,d)$ is isometrically 4-embeddable in $\ell^3_2$ if and only if  for all set  $\{c_1,c_2,c_3,c_4\}$ of real numbers and all finite sets $\{x_1,x_2,x_3,x_4\}$ of points in $X$, the implication
	\begin{equation}
	\label{eqn_neg_semi_def}
	\sum_{i=1}^{4} c_i=0\Rightarrow \sum_{i,j=1}^{4} c_i c_j d(x_i, x_j)^2\leq 0
	\end{equation}
	holds.
\end{lemma}

\subsection{Jensen-Shannon Distance}

\begin{lemma}[Tops{\o}e \cite{fuglede2004jensen}]
For an appropriate definition of Jensen-Shannon divergence (JSD), the space $(M_+^1(A),\sqrt{\textit{JSD}})$ is isometrically isomorphic to a subset in Hilbert Space.
\label{lemma_topsoe}
\end{lemma}
The term Jensen-Shannon divergence is used variously with slightly different meanings; to avoid ambiguity, we define it here as
\begin{equation*}
\mathit{JSD}(v,w) = 1 - \tfrac{1}{2}\sum_i (h(v_i) + h(w_i) - h(v_i + w_i) )
\end{equation*}
where
\begin{equation*}
h(x) = -x \log_2 x
\end{equation*}
which formulation,  explained in \cite{connor:jsd}, is consistent with other authors and neatly bounds the range into [0,1].

Here, the set $M_+^1(A)$ is the set of probability distributions, which we can safely interpret as a set of positive numeric vectors $\{v\} \in \mathbb{R}^n$ for some $n$ where $\sum_{i}^n v_i = 1$  (although the original definition extends to continuous spaces as well.)  Tops{\o}e uses Schoenberg's conjecture to prove this property by showing that JSD is itself a negative semidefinite mapping with the semi-metric properties. Although it has already been proved by more than one author that Jensen-Shannon distance (with the meaning of $\sqrt{\textit{JSD}}$ in Tops{\o}e's notation) is a proper metric (\cite{endres:2003},\cite{OstVaj03}) this proof of Hilbert space embedding gives that as a rather more elegant side-effect.

\begin{theorem}
The space $(M_+^1(A),\sqrt{\textit{JSD}})$  is isometrically $4$-embeddable in \threed, and can therefore use Hilbert Exclusion with hyperplane partitioning.
\end{theorem}
This is now  a direct consequence of Lemmata \ref{lemma_scholtes} and  \ref{lemma_topsoe}.

\subsection{Triangular Distance}
\label{subsec_triangle}
To establish the generality of our results, we give one more example of a proper metric which is also Hilbert space embeddable and can therefore  be indexed using Hilbert Exclusion.

The function 
\[
k(v,w) = \sum_ i \frac{\, (v_i - w_i)^2}{v_i + w_i} 
\]
(where $v,w \in \mathbb{R}^n, \sum_i v_i = \sum_i w_i = 1$) 
has been identified and named in \cite{Topsoe2000} as \emph{Triangular Discrimination}. Although rarely used in pratice, it is of significant interest as it has relatively tight upper and lower bounds over  the much more expensive Jensen-Shannon distance \cite{Topsoe2000}. $k$ is  a semi-metric, so if it is negative semidefinite then $\sqrt k$ is a Hilbert-embeddable proper metric.

As $k$ is a summation it is sufficient to prove that
\[f(x,y)=\dfrac{(x-y)^2}{x+y}\]
 is conditionally negative semidefinite. 
Recalling the definition of negative semidefinite (Equation \ref{eqn_neg_semi_def}) we require
\[\sum_{i,j}  \frac{(x_i - x_j )^2}{x_i + x_j}c_ic_j \le 0\]
for any finite set of real numbers $(c_i)_{i\le m}$ such that $\sum_i c_i = 0$
and for any finite set $(x_i)_{i\leq n}$ of points in $X$.
 
Observing that $(x_i-x_j)^2=(x_i+x_j)^2-4x_ix_j$  we obtain

\begin{align*}
\sum_{i,j}^{m} c_i c_j \dfrac{(x_i-x_j)^2}{x_i+x_j}
 &=\sum_{i,j}^{m} c_i c_j x_i + \sum_{i,j}^{m} c_i c_jx_j - 4\sum_{i,j}^{m} c_i c_j\dfrac{x_ix_j}{x_i+x_j}\\
 %
 %
 & =-4 \sum_{i,j}^{m} c_i c_j\dfrac{x_ix_j}{x_i+x_j}
 \end{align*}
as the first two terms sum to zero.
 Thus it is sufficient to prove that \[\sum_{i,j}^{m} c_i c_j\dfrac{x_ix_j}{x_i+x_j}\geq 0\]

As the index $i,j$ such  $x_i=0$ or $x_j=0$ do not  contribute to the summation, we can assume that all the $x_i, x_j$ are positive.
  \begin{align*}
     \sum_{i,j}^{m} c_i c_j\dfrac{x_ix_j}{x_i+x_j}&=\sum_{i,j}^{m} c_i c_j{x_ix_j}\int_{0}^{\infty} e^{-t(x_i+x_j)} dt\\
     &=\int_{0}^{\infty} \sum_{i,j}^{m} c_i c_j{x_ix_j}e^{-t(x_i+x_j)} dt\\
      &=\int_{0}^{\infty} \left(\sum_{i}^{m} c_i {x_i}e^{-tx_i}\right)\left(\sum_{j}^{m} c_j {x_j}e^{-t x_j}\right) dt\\
      &=\int_{0}^{\infty} \left(\sum_{i}^{m} c_i {x_i}e^{-tx_i}\right)^2 dt\geq 0
  \end{align*}

%
%
%
%
%
%
 This therefore gives us that
\[D_{\text{tri}}(v,w) = \sqrt{\sum_ i \frac{\, (v_i - w_i)^2}{v_i + w_i}}\]
which we name as Triangular Distance, is a proper metric such that $(M_+^1(A),D_{\text{tri}})$  is a metric space which is isometrically embeddable in Hilbert space.
 
\subsection{Spaces with Cosine Distance}
The term ``Cosine" distance does not have a unique meaning in the metric space literature and so requires an explanation.

It has long been known that, for two values $v,w$ in $\mathbb{R}^n$, then the function
\[
S_{\text{Cos}}(v,w) = \frac{v \cdot w}{\|v\|\|w\|}
\]
gives a convenient estimate of their dimensional correlation. One  advantage of this is that it is cheap to calculate, especially when the space is sparse such as applications in information retrieval. This function calculates the cosine of the angle between the vectors, and is best referred to as the Cosine Similarity Coefficient.

As it is bounded in $[0,1]$, the function $f(v,w) = 1 - S_\text{Cos}(v,w)$ gives a bounded divergence coefficient;  however this function is  not a proper metric, as it lacks triangle inequality. A function which gives the same rank order and is also a proper metric can be simply achieved by converting this value into the angle between two vectors, which can be caused to range within $[0,1]$ by $d_\text{Cos}(v,w) = 1 - cos^{-1}(S_\text{Cos}(v,w)) / 2\pi$. In the metric space literature, this function is sometimes referred to as Cosine Distance \cite{SISAP_man,connor:multivariate}.

This function is a proper metric, but is not isometrically embeddable in Hilbert space. However, there exists another rank-equivalent function based on the Cosine similarity:
\[
d_\text{Cos}(v,w)  = \sqrt{1-S_\text{Cos}(v,w)}
\]
In fact, since $\|v-w\|^2= \|v\|^2+\|w\|^2-2v\cdot w$, the distance $d_\text{Cos}(v,w)$ is equivalent to the Euclidean distance computed on the normalized vectors ${v}/{{\|v\|}}$ and ${w}/{{\|w\|}}$:
\[
d_\text{Cos}(v,w)  = d_\text{Cos}\left(\frac{v}{ {\|v\|}},\frac{w}{{\|w\|}}\right)=
\frac{1}{\sqrt{2}}\,{\left\|\frac{v}{ {\|v\|}}-\frac{w}{ {\|w\|}}\right\| }
\]
and is 
therefore isometrically 4-embeddable in three dimensional Euclidean space, and hence in a Hilbert space.
\subsection{High-Dimensional Euclidean Space}

For completeness we reconsider $n$-dimensional Euclidean space for any $n$ in the context of Hilbert embedding. From  Lemmata \ref{lemma_schoenberg}  and \ref{lemma_scholtes}  it is sufficient to show that the function $K(v,w) = \sum_i (v_i - w_i)^2$ is a conditionally negative semi-definite semi-metric, which is straightforward to demonstrate using a similar proof to that used in Section \ref{subsec_triangle}.

\subsection{Non-Embeddable Spaces}
\begin{figure}[t]
\begin{center}
\begin{tikzpicture}[scale=1]

\filldraw (0.5,0.6) circle (1pt)  node[anchor=east]{p};
\filldraw (0,1) circle (1pt)  node[anchor=east]{q};
\filldraw (0.5,0) circle (1pt)  node[anchor=east]{r};
\filldraw (1,1) circle (1pt) node[anchor=east]{s};

\draw[] (0.5,0.6)-- (0,1);
\draw[] (0.5,0.6) -- (0.5,0);
\draw[] (0.5,0.6)-- (1,1);

\end{tikzpicture}
\hspace{1cm}
\begin{tikzpicture}[scale=1]

\filldraw (0,0) circle (1pt) node[anchor=east]{a};
\filldraw (0,1) circle (1pt) node[anchor=east]{b};
\filldraw (1,0) circle (1pt) node[anchor=west]{d};
\filldraw (1,1) circle (1pt)node[anchor=west]{c};

\draw[] (0,0) -- (0,1);
\draw[] (0,0) -- (1,0);
\draw[] (1,0) -- (1,1);
\draw[] (0,1) -- (1,1);
\end{tikzpicture}
\caption{Non-Embeddable Metric Spaces}
\end{center}
\label{fig_non_embed_graphs}
\end{figure}
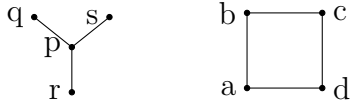
\begin{table}[t]
\label{table_non_embeddable}
\begin{center}
\caption{Classic non-embeddable examples: distances between nodes for the star graph and the Hamming Cube}{
\begin{tabular}{|c|c|c|c|c|}
\hline
&q&r&s&p						\\
\hline
q&	&	2&	2& 1					\\
\hline
r&	&	&	2&1					\\
\hline
s&	&&&1					\\
\hline
p&&&&						\\
\hline
\end{tabular}
\qquad
\begin{tabular}{|c|c|c|c|c|}
\hline
&a&b&c&d						\\
\hline
a&	&	1&	2& 1					\\
\hline
b&	&	&	1&2					\\
\hline
c&	&&&1					\\
\hline
d&&&&						\\
\hline
\end{tabular}}
\end{center}
\end{table}%

To complete the picture, it is worth mentioning that not all metric spaces are 4-embeddable in \threed {}; it is therefore necessary to make a proper assessment of the space in question before using Hilbert Exclusion.

Figure \ref{fig_non_embed_graphs} shows two example of  graphs (taken from \cite{MatouzekDimRedDist}) where the distance between two nodes is defined as the minimum number of paths that must be traversed. This is a proper metric; the node-to-node distances are given in Table \ref{table_non_embeddable}. It is immediately apparent that neither of these sets of four points are isometrically 4-embeddable in  \threed.

For the star graph, consider $p$ as the centre of a sphere on which the other points lie; however as the maximum distance between any two points on a sphere of radius 1 is 2, which occurs only when they lie at either end of a diameter, then  no such three points can exist in three dimensions.
 
 Similarly for the Hamming cube; if the diagonals are fixed at length 2, then at least one of the sides much have a length of no less that \smash{$\sqrt 2$}.

Some common distances, for example Chebyshev and Manhattan  distances, are not Hilbert embeddable. It is straightforward to show that Manhattan distance itself is conditionally negative semi-definite, and therefore the space \smash{$(\mathbb{R}^n,\sqrt M)$}, where $M$ is the Manhattan  distance is Hilbert embeddable. 
More  generally, in \cite{blumenthal1953} it is proved that if $(X,d)$ is a metric space than $(X, d^\alpha)$, with $0\leq \alpha\leq 1/2$, is isometrically 4-embeddable in $\ell^3_2$ and so $(X, d^\alpha)$ can be searched using the Hilbert Exclusion.
 However  for practical purposes the advantages of using Hilbert Exclusion are likely to be outweighed by a huge increase in intrinsic dimensionality.

Levenshtein distance, used for example in text processing and computational biology, is well known to be a proper metric.
In \cite{toth2004handbook} it is stated 
\begin{quote}
``not much is known about embeddability  of  this metric in normed spaces \dots  It is known however that the Levenshtein metric, restricted to a certain set of strings, is isomorphic to the shortest path metric over $K_{2,n}$"
\end{quote}
Therefore Levenshtein distance is not isometrically embeddable in a Hilbert space.
\section{Analysis}
\label{sec_performance_analysis}
As Hilbert Exclusion  is strictly weaker than Hyperbolic Exclusion, the performance of any partition-based indexing mechanism is always better. The distance between the pivot points is required as well as the distance between each pivot and the query, however this may always be calculated during the building of any indexing structure and adds nothing to the cost of a query. Query evaluation cost is totally dominated by the number of dynamic distance calculations required and the use of memory where the objects are large; the minor increase in arithmetic cost, and the extra space required to store the distance between pivots, do not make any significant difference to the query cost.

The many different index mechanisms reported show that performance is highly dependent on many factors, not least the cost of a distance calculation, the size of the objects, and other factors including the intrinsic dimensionality and the distribution of the data within the space. Furthermore most of the more sophisticated mechanisms use a mixture of hyperplane and cover radius exclusion; it may be that enhanced performance of hyperplane exclusion could make a significant difference to the choice of index. It is not therefore possible to analyse a simple ``performance improvement" in general terms.

We therefore give analysis of the improved exclusion condition as follows.

\begin{enumerate}
\item Exclusion power: for a given finite space, we randomly select pairs of pivot points that partition a space into two halves. The exclusion power of each  mechanism can then be measured as the probability of a randomly-selected query being able to avoid searching either half of the space based only on its distance  from the two points. This is always greater for Hilbert Exclusion than for Hyperbolic Exclusion; in Section \ref{subsection_power} we give figures for  various spaces.

\item Improvement: for a given metric space, simple data structures relying primarily on hyperplane partitioning are built, namely a generalised hyperplane tree and a monotonous hyperplane tree. The same index structures can be used with either Hilbert or Hyperbolic exclusion; improvement is measured as a simple multiplicative factor between the two. We give results in Section \ref{subsection_jmprovement}.

\item Real-world data: The SISAP forum%
\footnote{www.sisap.org}  publishes a number of large data sets drawn from real world contexts which are commonly used as benchmarks for different indexing mechanisms. Results over these have been reported for many different indexing mechanisms. We take the best of these mechanisms, which uses both radius and hyperplane exclusion, and compare it using Hyperbolic and Hilbert exclusion mechanisms. Results for this are given in Section \ref{subsection_real_world}.

\end{enumerate}

\subsection{Experimental Method}

Any exclusion mechanism works well within a context of low dimensionality spaces and small query thresholds. To give a general overview of the tradeoffs, we perform all tests over a variety of spaces and thresholds.

 In all cases, we generated pseudo-random data sets of one million elements within the unit hypercube, evenly distributed within each dimension, within $\mathbb{R}^d$ for $d \in \{6,8,10,12,14\}$. In the results presented we name the spaces used based on the metric and the number of Cartesian dimensions, eg {\small $\mathtt{euc\_10}$} for Euclidean distance over $\mathbb{R}^{10}$, {\small $\mathtt{jsd\_12}$} for Jensen-Shannon distance%
 \footnote{for {\small $\mathtt{euc}$} and {\small $\mathtt{tri}$}, each point is normalised so that $\sum_i v_i = 1$ }
  over $\mathbb{R}^{12}$ etc.

Search thresholds were derived by experiment, for each space, as those which would return around $n$ results per million data, for $n \in \{1,2,4,8,16,32\}$.

For each space we also calculated the Intrinsic Dimensionality (IDIM, \cite{Chavez:2001}), generally believed to give a good ``rule of thumb" impression of how tractable a space is to metric indexing techniques; folklore indicates that spaces with an IDIM of greater than around 6 are challenging, and those with an IDIM of greater than about 10 are intractable\footnote{There is no very clear scientific evidence for this that we know of, but the opinion is widely held among researchers at venues such as SISAP}. IDIM is defined over a sample of distances calculated over randomly selected points from within the space, based on the mean $\mu$ and standard deviation $\sigma$ of these distances, as \smash{$\frac{\mu^2}{2 \sigma^2}$}.

Table \ref{table_idims_thresholds} in Appendix \ref{section_appendix_1} gives values for IDIM and thresholds calculated for each space. Given these values, all experimental results  are  obtainable through repetition of the experiments described. All results are independent of the computer upon which they are performed, and all figures presented represent mean values where experiments were repeated until the standard error of the mean was less than 1\% of the value given.

\subsection{Exclusion Power}

\label{subsection_power}

\begin{figure}
\centering
\fbox{\includegraphics[width=0.7\columnwidth]{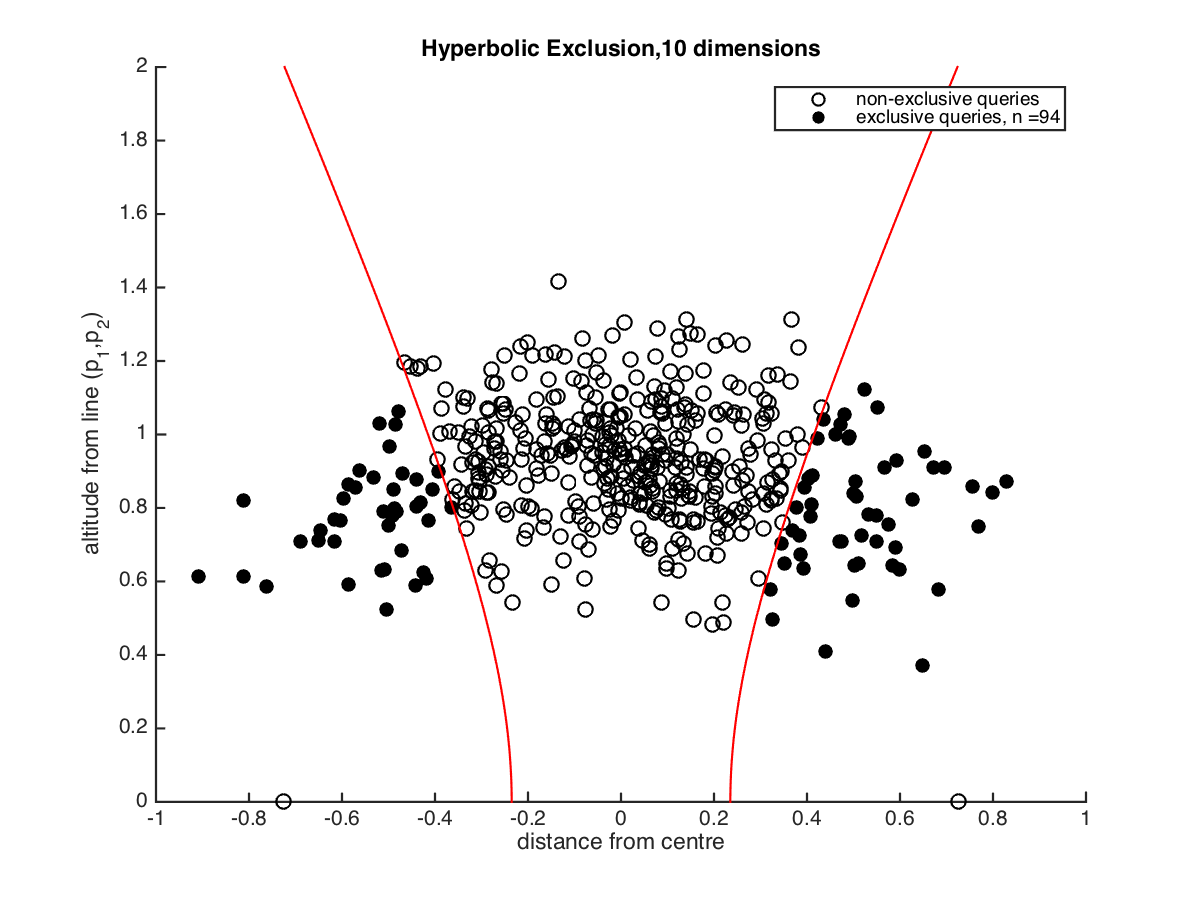}}

\caption{Hyperbolic Exclusion. Two points are chosen at random from a finite space and placed symmetrically on the X axis, either side of the origin, separated by the distance between them in the original space. The remaining points are plotted in the upper half of the space according to their distance from these two points. Relative distances among these points are not significant as each point represents a different embedding function. Those coloured solidly are those which, were they queries, would allow the semispace on the opposing side to be excluded from a search.}
\label{fig_illustrative_voronoi}
\end{figure}

\begin{figure}
\centering
\fbox{\includegraphics[width=0.7\columnwidth]{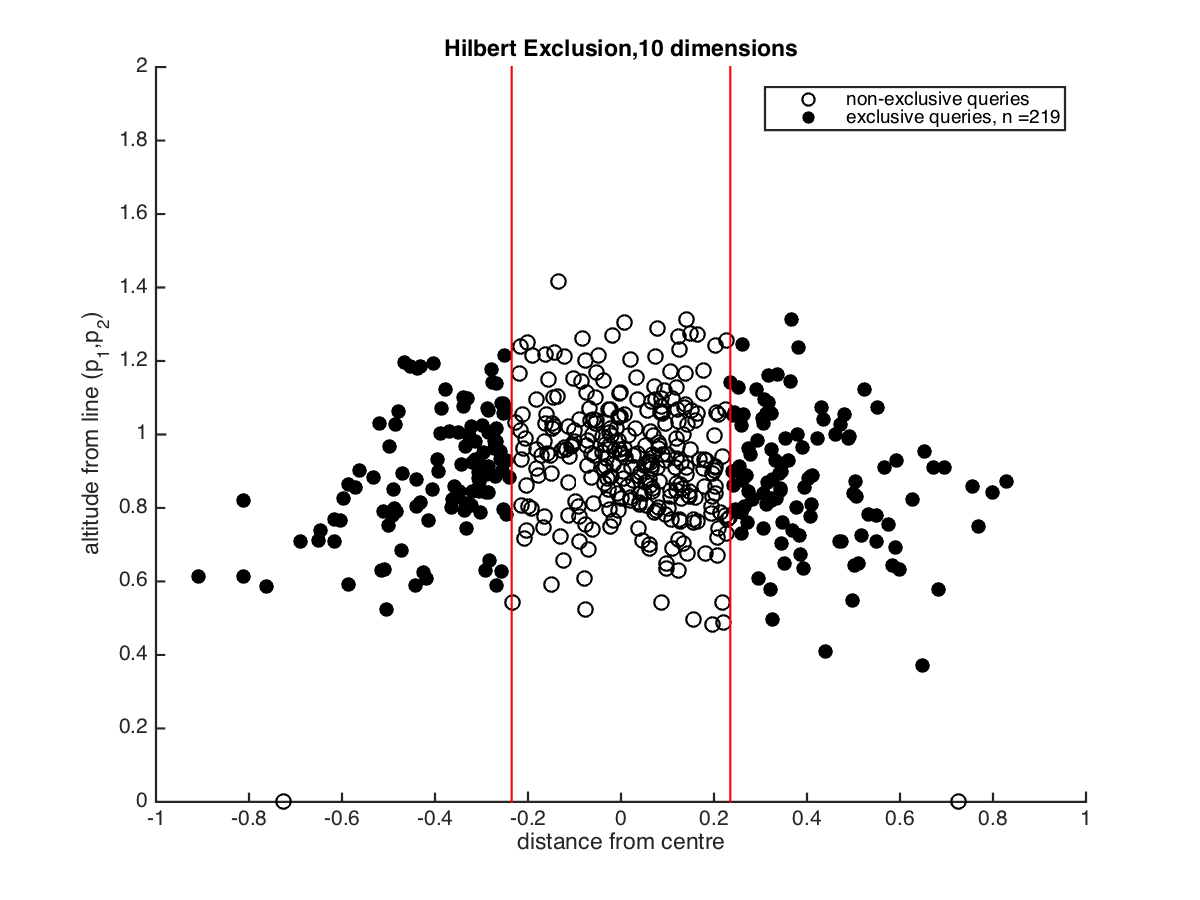}}

\caption{Hilbert Exclusion. The same plot as in Fig \ref{fig_illustrative_voronoi}; the solidly coloured points  represent queries that allow the opposing semispace to be excluded using Hilbert Exclusion. These are now all points at least the threshold distance from the separating hyperplane, which includes many more queries for the same threshold. }
\label{fig_illustrative_cosine}
\end{figure}

\begin{figure}[htbp]
\centering
\fbox{\includegraphics[width=0.7\columnwidth]{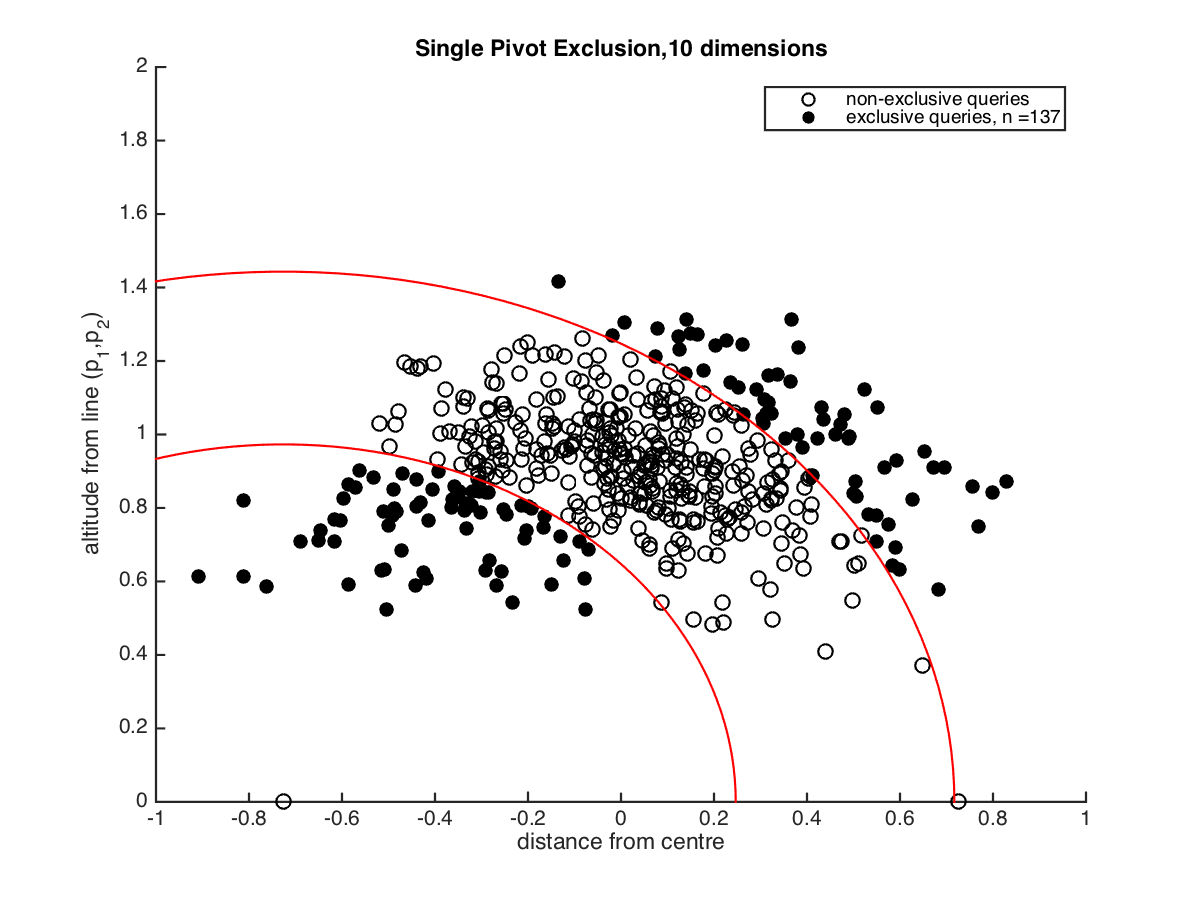}}

\caption{Pivot Exclusion. The same plot as Figs \ref{fig_illustrative_voronoi} and \ref{fig_illustrative_cosine}, but now the left-hand point on the  X axis is used to exclude queries based on distance from that alone. Semispaces are defined according to the median distance from this point, and solidly coloured points indicate those whose distance from the pivot point is more than the query threshold away from this median. }
\label{fig_illustrative_pivoting}
\end{figure}

Figures \ref{fig_illustrative_voronoi}, \ref{fig_illustrative_cosine} and \ref{fig_illustrative_pivoting} illustrate the exclusion power test. Each figure shows the same set of 500 randomly generated points in a 10-dimensional  Euclidean space. A futher two points are also generated to act as pivots.

In Figures \ref{fig_illustrative_voronoi} and \ref{fig_illustrative_cosine}, the distance between the pivot points is measured as $d$; an embedded 2D plane is then constructed with these points at $(0,-1/2d)$ and $(0,1/2d)$ respectively. Each point in the generated set is then measured against these two points, and plotted in the upper half of the plane according to these distances. It can be seen that the same points are plotted in both figures. Note that the relative distances within the plot are of no significance; each point represents a different embedding function. However the position of each point within the space is individually significant with respect to the pivot points.

A query radius is chosen, in this case one that would be expected to return around one point per million from a large set. Figure \ref{fig_illustrative_voronoi} highlights those points which satisfy  the Hyperbolic Exclusion condition, and Figure \ref{fig_illustrative_cosine} highlights those which satisfy  Hilbert Exclusion. As well as noting the number is substantially greater (201 against 75 in this example) it is instructive to note the shape of the exclusion zones within the two figures; Figure \ref{fig_illustrative_voronoi} clearly shows the shape of the hyperbola which demarcates the zone, whereas Figure \ref{fig_illustrative_cosine} clearly shows   parallel lines either side of the central axis.

To give a reference diagram for single pivot-based exclusion, Figure \ref{fig_illustrative_pivoting} gives the same plot but highlights those which are more than the same query threshold from the median distance to the left-hand pivot point, which are those that could be excluded according to radius-based exclusion from this point alone; there are 139 of these in this case.

In all spaces that we have measured, the single-pivot method has more exclusion power than Hyperbolic exclusion, but less power than Hilbert Exclusion. In metric indexing things are not this simple, as in particular hyperplane separation is normally used to effect in conjunction with  cover radius exclusion. The greater exclusion potential of Hilbert Exclusion requires two distance calculations, against a single calculation for pivot-based exclusion; however many indexes have ways of amortising this extra cost. Finally, plane partitioning is very effective when the space is amenable to geometric separation, as it tends to cluster subsets which are relatively closer to each other, whereas ball partitioning tends to be less effective in this respect.

 In all  there is a hint that, when applicable, the new condition appears to enjoy the best of all worlds in this respect; at least it may make a significant difference to the choice of mechanism for a given data set, and may possibly inspire new mechanisms to be developed.

\subsubsection{Results}

 Table \ref{table_power_test_results} in Appendix \ref{section_appendix_exc_pow} gives outcomes of the exclusion power test for the three given Hilbert-embeddable metrics over spaces of various dimensions, using various query thresholds. These results are graphically summarised in Figure \ref{fig_power_improvement} for Euclidean spaces; the other two metrics give very similar patterns. The left-hand figure shows the exclusion percentage obtained at various dimensions and thresholds; it can be seen that Hilbert Exclusion performs much better than Hyperbolic Exclusion, and is much more tolerant to increases in both dimensionality and query threshold; that is, it performs relatively better as the space becomes less tractable.
 
 The right hand graphs illustrates this in terms of improvement of Hilbert over Hyperbolic exclusion, which again can be seen to increase sharply as the space becomes less tractable.

\begin{figure}[!t]

\makebox[\textwidth]{
\fbox{\includegraphics[width=0.45\columnwidth]{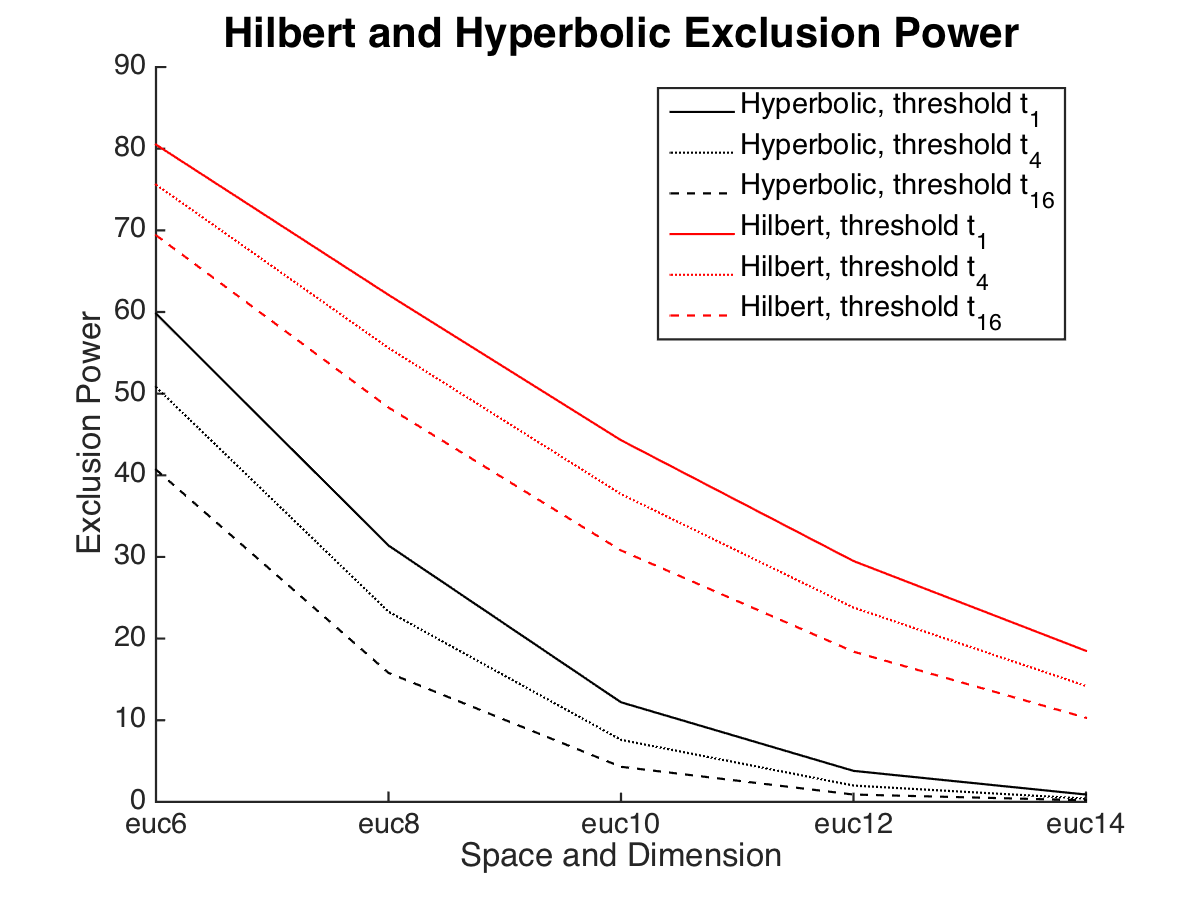}}
\fbox{\includegraphics[width=0.45 \columnwidth]{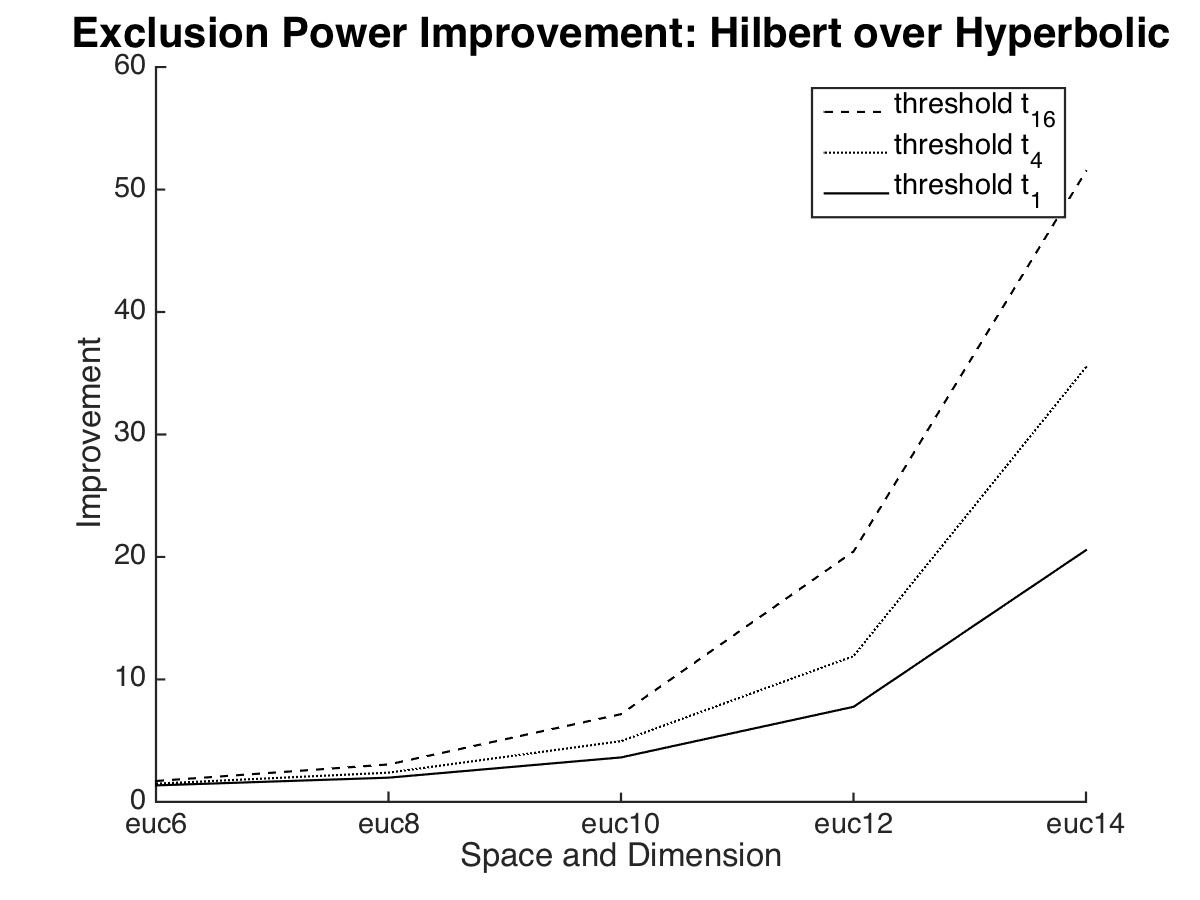}}
}

\caption{Exclusion Power tests: Each figure shows five different dimensionalities, and three different search thresholds, for Euclidean  spaces. Jensen-Shannon and Triangular spaces  gives similar results. Left figure is percentage exclusion, right is relative improvement of Hilbert over Hyperbolic. }
\label{fig_power_improvement}
\end{figure}

\subsection{Improvement}

\label{subsection_jmprovement}

To give  a more  practical measurement of performance improvement, the two exclusion mechanisms have also been tested over metric indexes built over actual data sets. The indexes used are the general hyperplane tree  (GHT, \cite{GHT}) and the monotonous hyperplane tree (MHT, \cite{Noltemeier1992})%
\footnote{Originally named the ``Monotonous Bisector* Tree"}%
, which are in a sense the most ``pure" (and certainly the simplest) hyperplane indexing structures.  In these experiments, for each data set used the same data structure is created, the only difference is in the exclusion mechanism used.

It should be noted here that  the notions of ``bisector" and ``hyperplane" tree are conceptually different; although they share the same construction algorithm, bisector trees use a cover radius for pivot-based exclusion, and hyperplane trees use, normally, hyperbolic exclusion. In our experiments we use both cover radius and hyperplane exclusion mechanisms, as would be normal in practice, and compare the use of hyperbolic exclusion with Hilbert exclusion.

\subsubsection{Results}
Table \ref{table_query_results} in Appendix \ref{section_appendix_exc_pow}  shows, for various metrics and dimensionalities, the cost of indexing two  hyperplane-based metric index structures with the different exclusion strategies. 
Figure \ref{fig_performance} shows some of the results in graphical form.

 It can be seen that, for all spaces, Hilbert Exclusion always gives better performance than Hyperbolic Exclusion; this is expected, as the exclusion condition is strictly weaker. Table \ref{table_query_results} shows that, under Hyperbolic Exclusion,  the MHT always gives marginally improved performance over the GHT; again, this is already known and understood. It can also be seen that the GHT under Hilbert Exclusion  gives equal or better performance than the MHT under Hyperbolic Exclusion.  Interestingly however, the improvement given by using Hilbert Exclusion over the MHT  is  dramatically better than the improvement given over the GHT, for which we do not currently have a reason.

Another interesting observation is shown on the right of Figure  \ref{fig_performance}, which gives the ratio of the number of distances calculated by the MHT for the two exclusion mechanisms; it can be seen that, for all search thresholds, this reaches a maximum at around 10 dimensions and then decreases again. This can be explained by the fact that, for very tractable spaces,both mechanisms function very well; there is not therefore a great improvement. For intractable spaces, neither mechanism can do well and so again the relative improvement becomes less. The observation is in keeping with the left hand diagram shown in Figure \ref{fig_power_improvement}, where it be seen that the gap in  exclusion power of the two mechanisms is greatest at around the same range of dimensions.

\begin{figure}[!t]

\makebox[\textwidth]{
\fbox{\includegraphics[width=0.45\columnwidth]{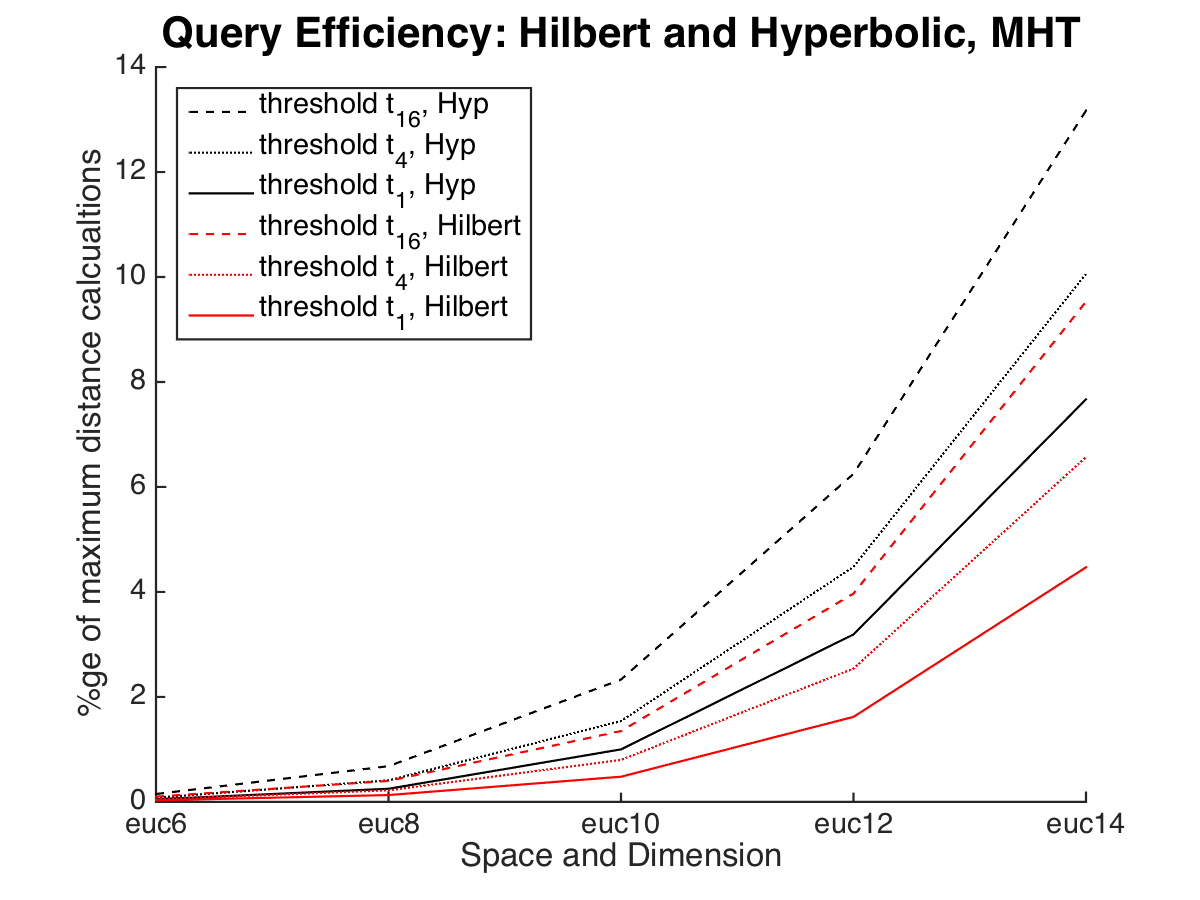}}
\fbox{\includegraphics[width=0.45 \columnwidth]{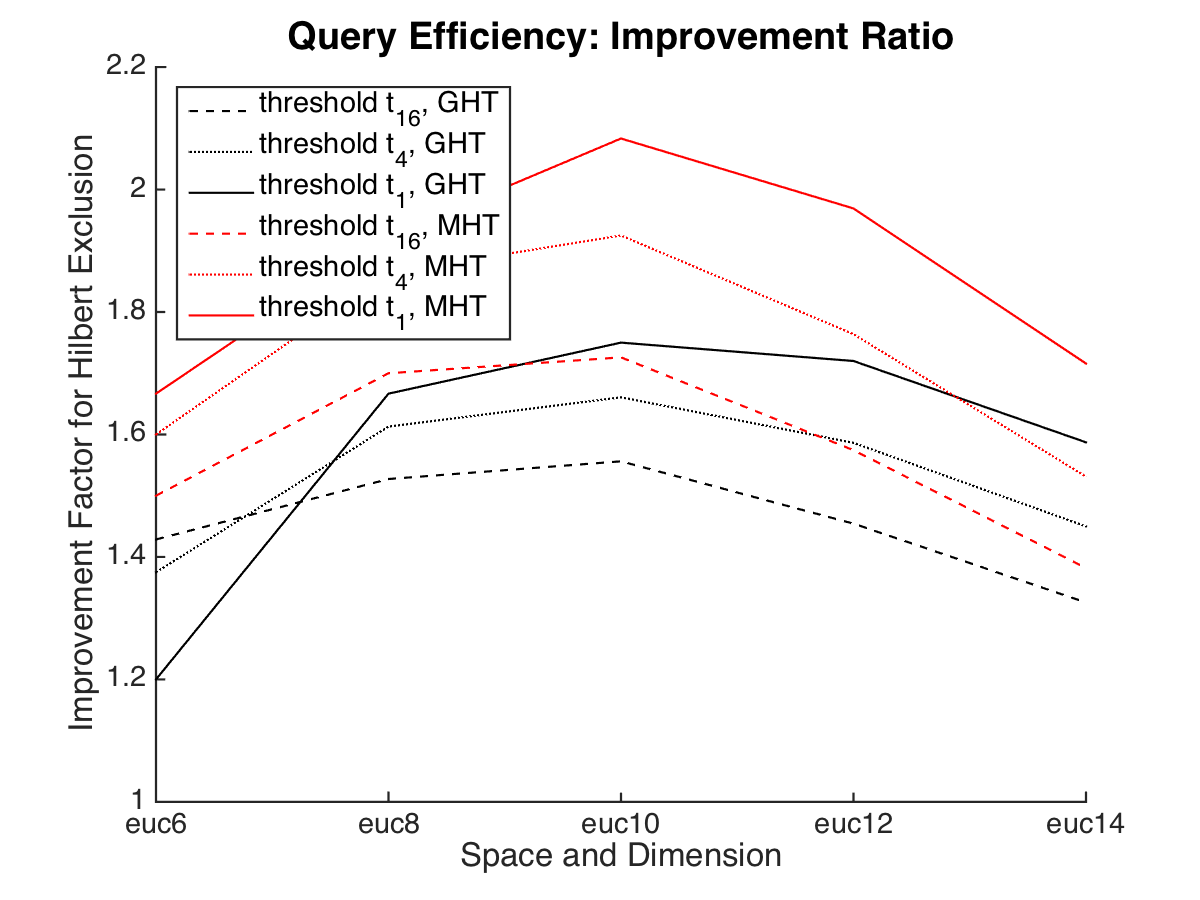}}
}

\caption{ The left hand graph shows absolute performance data for the MHT, at various dimensionalities and thresholds, for the two exclusion mechanisms. The right hand graph shows the same results interpreted as an improvement ratio, and also includes data from the GHT.  In the left hand graph, lines of the same pattern represent the same data, and the same index structures, only the query exclusion mechanism is different.}
\label{fig_performance}
\end{figure}

\subsection{``Real-world" data}

\label{subsection_real_world}

There are many different contexts for metric search, and no mechanism is generally believed to be best for all purposes. The most competitive comparator at the time of writing is the Distal Spatial Approximation Tree (DiSAT) \cite{dSatIS} which has been shown to perform better than a large range of other mechanisms. The authors write:

\begin{quote}
``Our data structure has no parameters to tune-up and a small memory footprint. In addition it can be constructed quickly. Our approach is among the most competitive, those outperforming DiSAT achieve this at the expense of larger memory usage or an impractical construction time."
\end{quote}

We can therefore take this mechanism as the state of the art in metric indexing, and as it uses hyperplane partitioning we can test the effect of applying Hilbert Exclusion against the Hyperbolic Exclusion with which it has been defined. In their publication, the authors test the DiSat very extensively and it is in almost all cases the best performing index.

The SISAP forum%
\footnote{www.sisap.org}  publishes a number of large data sets drawn from real world contexts which are commonly used as benchmarks for different indexing mechanisms, and results for the DiSAT were given with respect to these. We have implemented the DiSAT as described  in \cite{dSatIS} and measured the same  results over  Euclidean spaces; therefore we need only compare this structure with the two different exclusion criteria.

The same experimental context was used: the SISAP ``colors" and ``nasa" data sets are used to build instances of DiSATs. In each case ten percent of the data is used as queries over  remaining 90 percent of the set, at  threshold values  which return 0.01, 0.1 and 1\% of the data sets respectively.

Figure \ref{fig_sisap_sat_comparison} shows the outcome of these experiments. It is clear that using Hilbert exclusion greatly improves the performance. 

\begin{figure}[!t]

\makebox[\textwidth]{
\fbox{\includegraphics[trim=10mm 70mm 10mm 80mm, width=0.45\columnwidth]{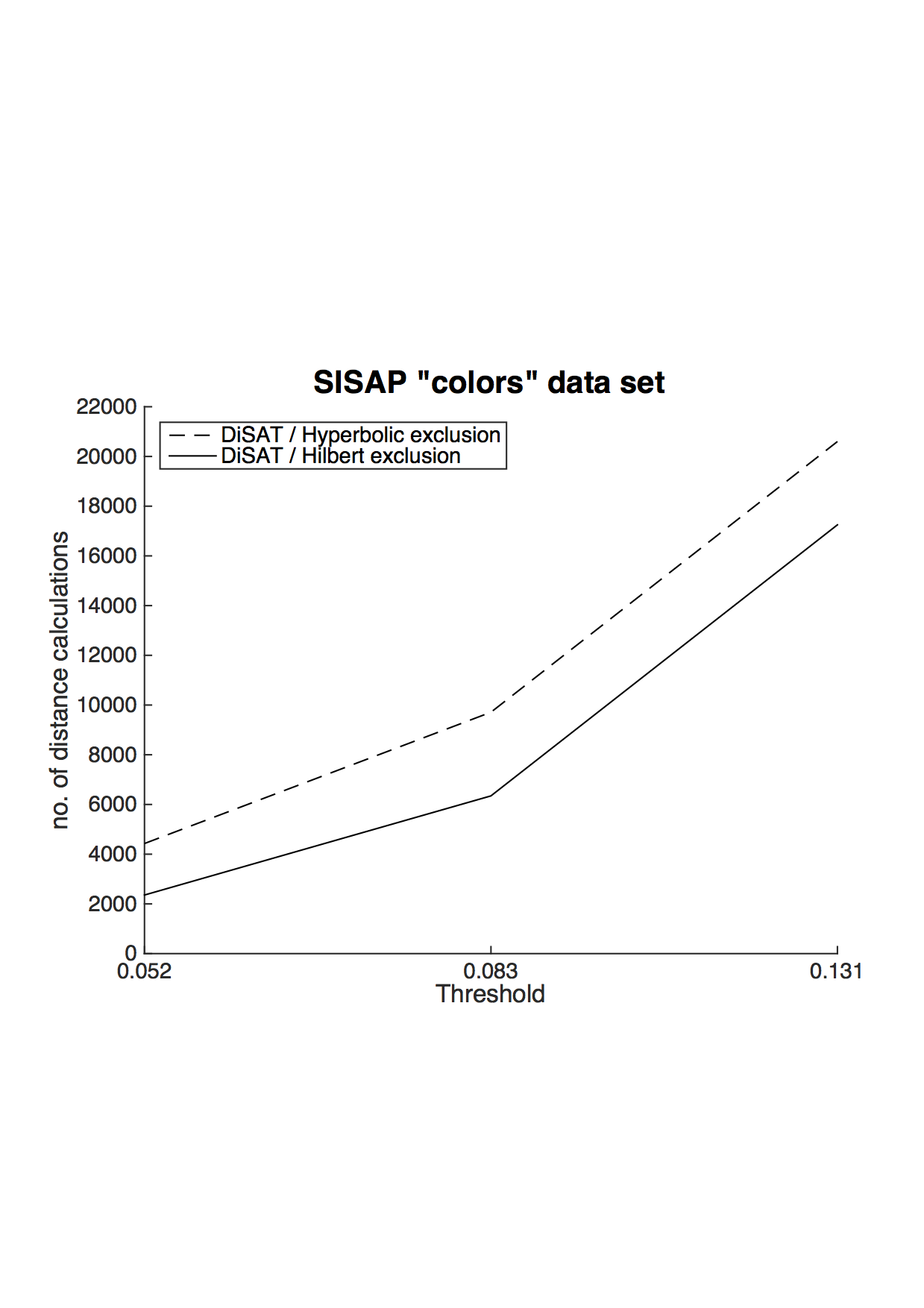}}
\fbox{\includegraphics[trim=10mm 70mm 10mm 80mm, width=0.45 \columnwidth]{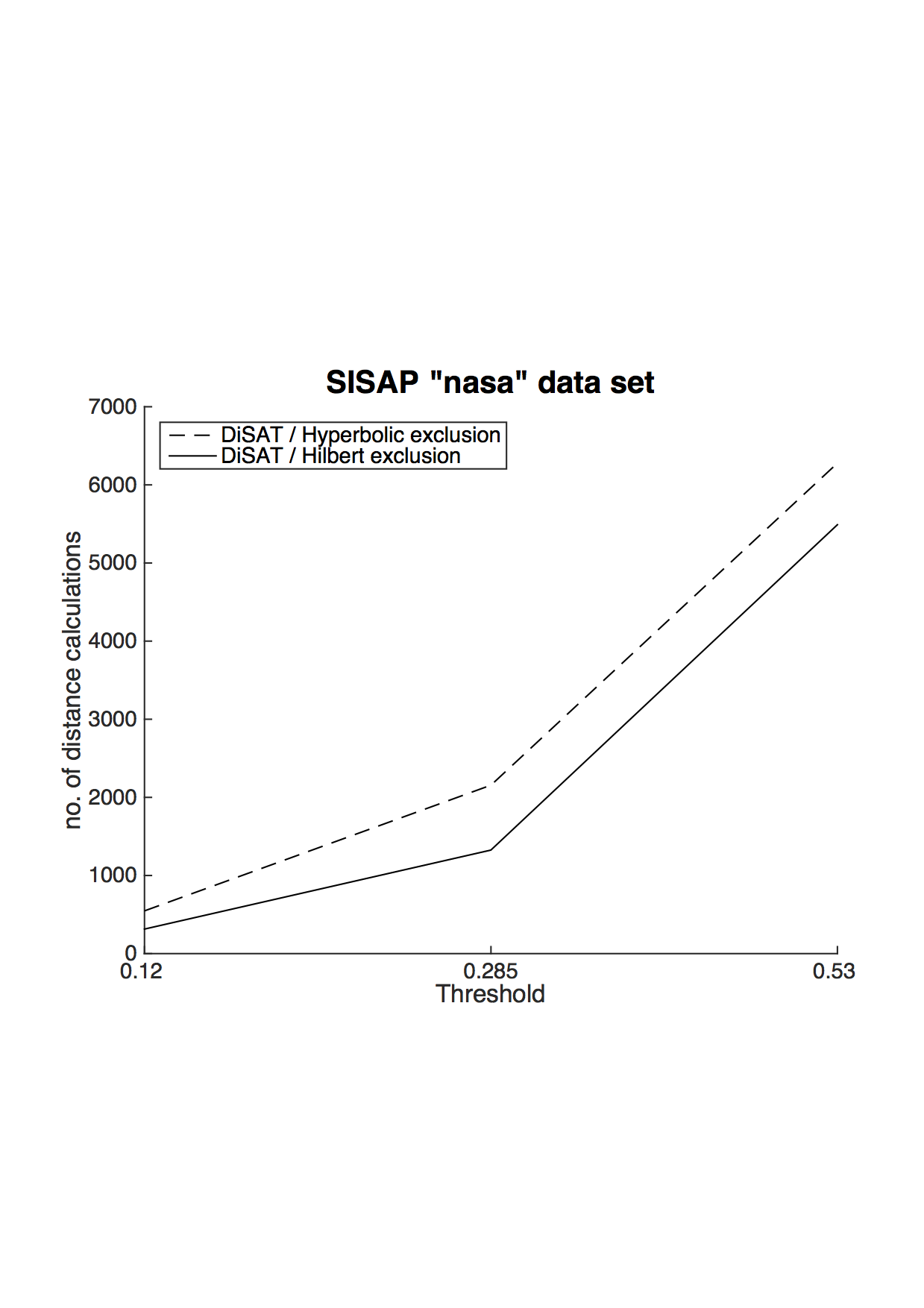}}
}

\caption{Comparing Hyperbolic and Hilbert Exclusion Conditions for the DiSAT. The two graphs represent benchmark applications of the ``state of the art"  DiSAT index over SISAP benchmark data sets, with significant improvements achieved through  changing only the exclusion condition. }
\label{fig_sisap_sat_comparison}
\end{figure}

\subsection{Correctness}
It is finally worth mentioning that during the course of the  experiments described in this paper, over one million queries have been executed over sets of at least one million data using a number of different indexes, including those using both Hyperbolic and Hilbert exclusion; all queries over the same sets, using different mechanisms, have been checked against each other and in all cases the results were identical. While we are confident about the correctness of the mathematical derivations given, it is nonetheless comforting to have such experimental validation.

\section{The Effects of Increasing Dimensionality}
\label{section_increasing_dimensionality}

The results given have shown how the relative advantage of Hilbert Exclusion over Hyperbolic Exclusion increases as the spaces become less tractable, that is as the intrinsic dimensionality increases.

A reason for this can be seen from studying the geometry of the two mechanisms in the three dimensional embeddings. As the dimensionality increases, there are three well-known effects: the mean distances between randomly sampled points increases;  the standard deviation  of these distances decreases,and  query thresholds greatly increase.
This last gives the greatest effect  in terms of the tractability of indexing mechanisms, and is an effect of the relative ratio of the volume of the unit hypercube and the unit hypersphere as dimensions increase. The volume of the unit hypersphere in $2k$ dimensions is \smash{$ \frac{\pi^k}{k!}$}, which decreases very rapidly after three dimensions, whereas  the volume of the unit hypercube remains as 1, independent of the dimension.

As can be seen from Table \ref{table_idims_thresholds}, in 6-dimensional Euclidean space the radius of a hypersphere with a volume of $10^{-6}$ is 0.076; in 14-dimensional Euclidean space it is 0.386. This has the effect of not only making the hyperbola wide, but also causing it to veer sharply away from the central hyperplane.

Figure \ref{fig_six_fourteen_dims} illustrates this effect by illustrating the siutation in both 6 and 14 dimensions for a small set of 500 randomly generated points in the unit hypercube.

\begin{figure}
\makebox[\textwidth]{
\includegraphics[width=0.6\columnwidth]{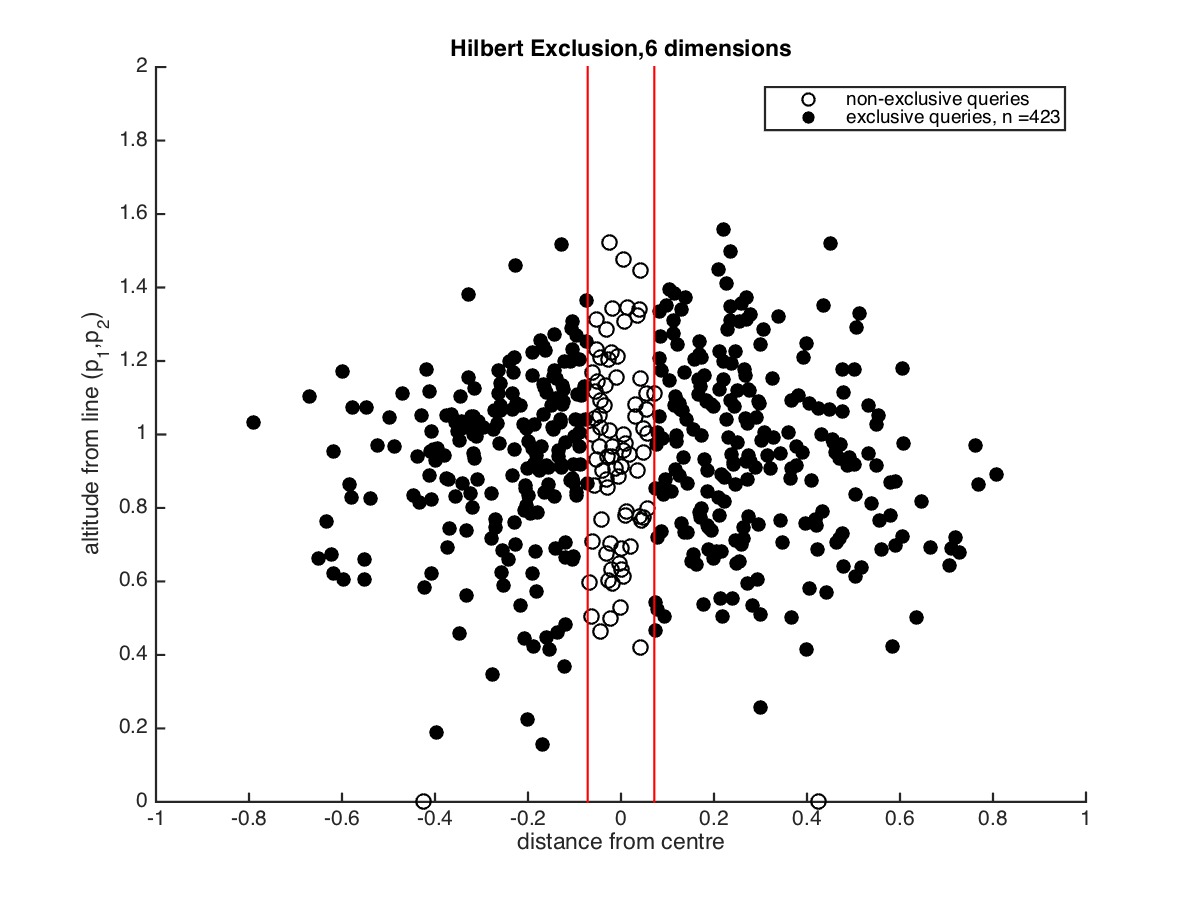}
\includegraphics[width=0.6\columnwidth]{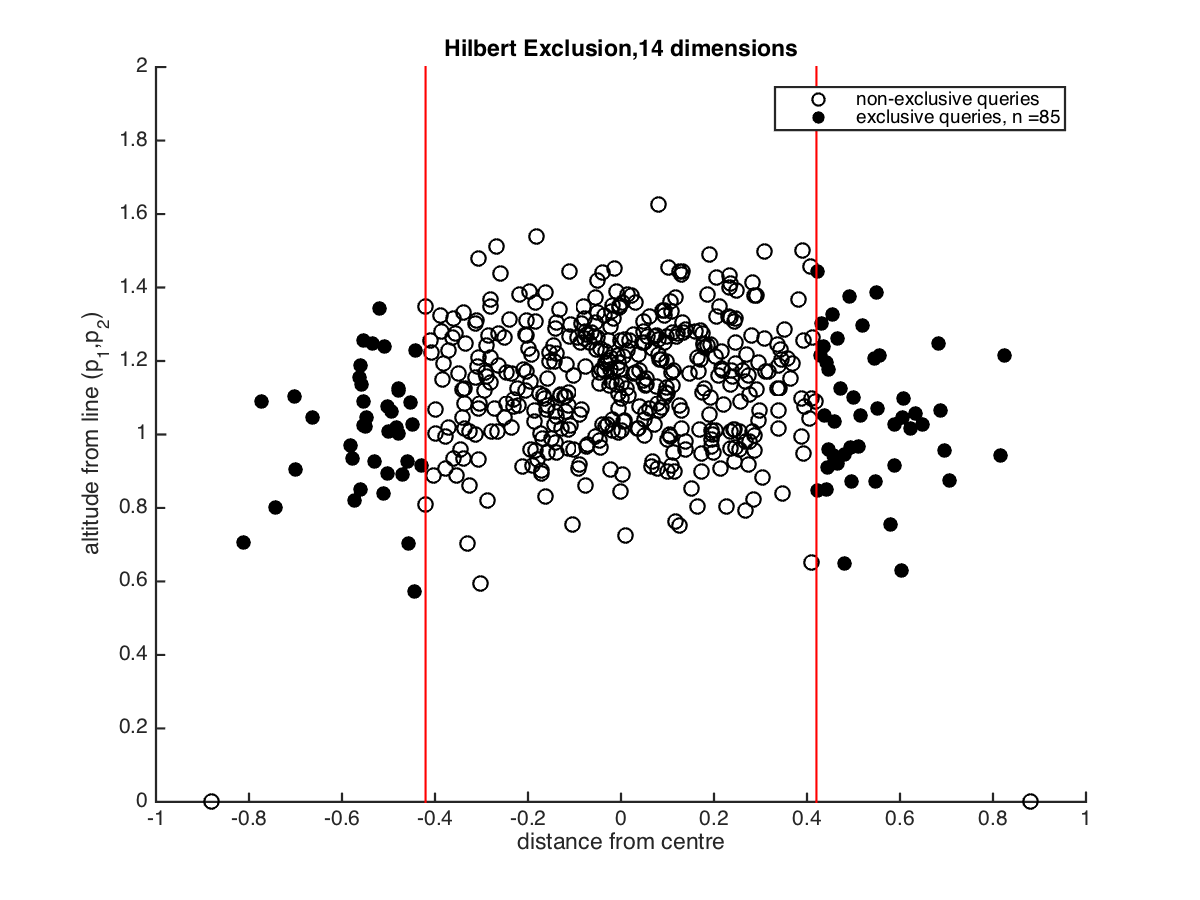}
}
\makebox[\textwidth]{
\includegraphics[width=0.6\columnwidth]{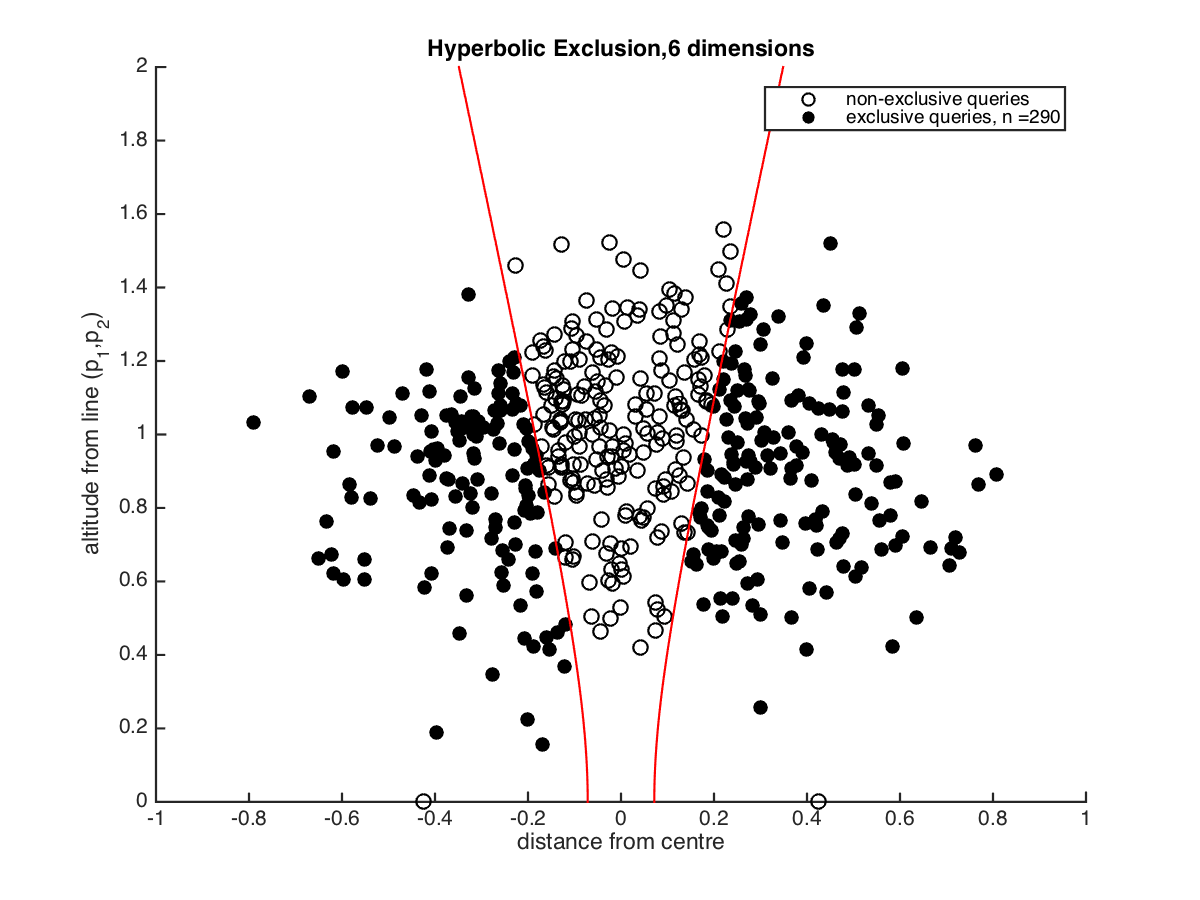}
\includegraphics[width=0.6\textwidth]{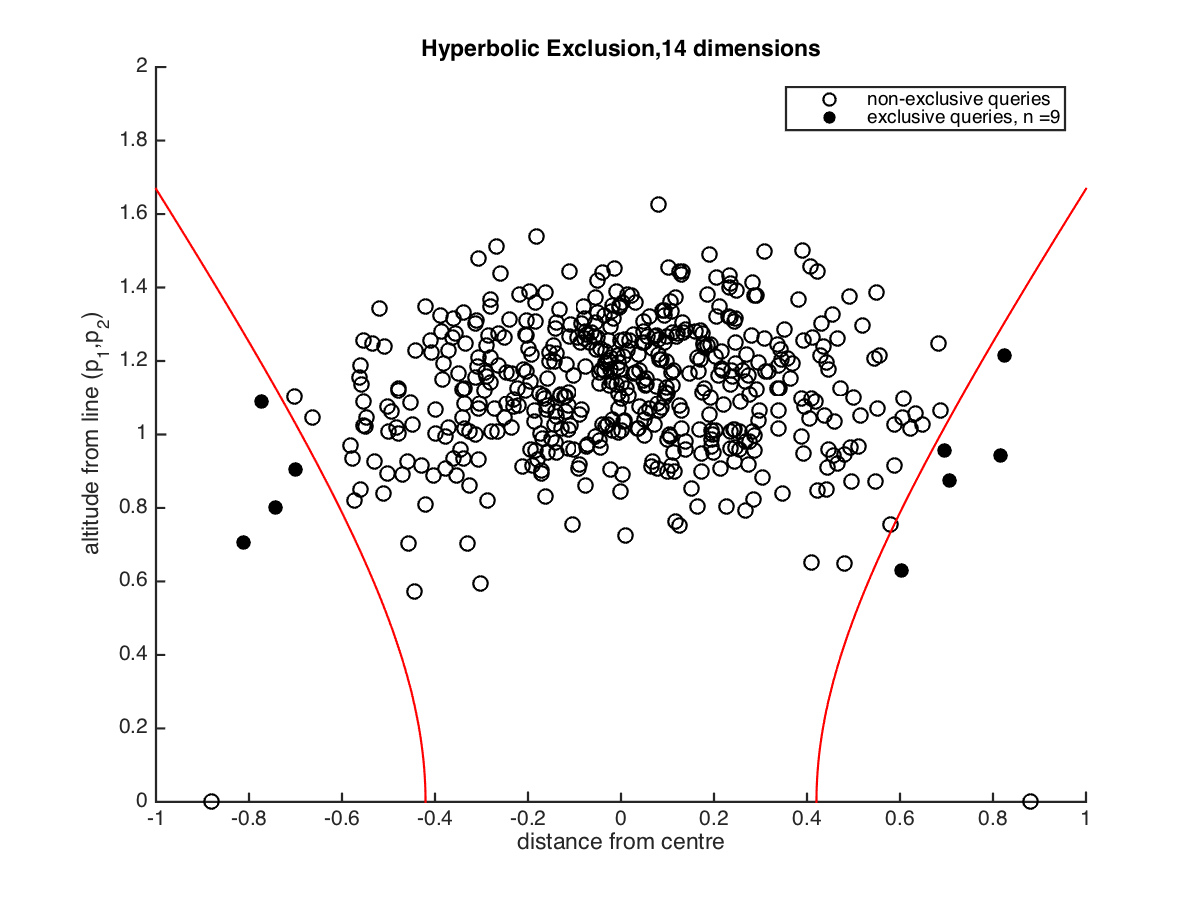}
}
\makebox[\textwidth]{
\includegraphics[width=0.6\columnwidth]{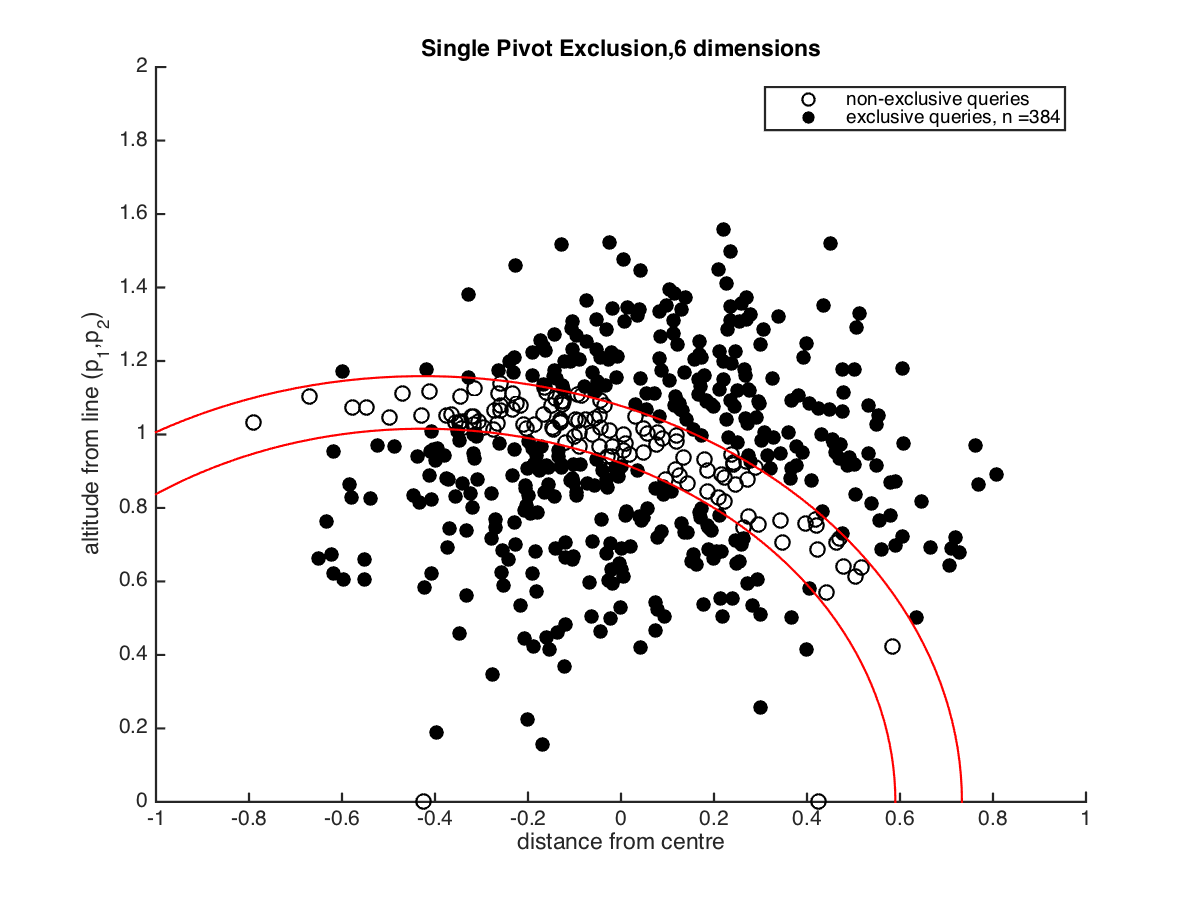}
\includegraphics[width=0.6 \columnwidth]{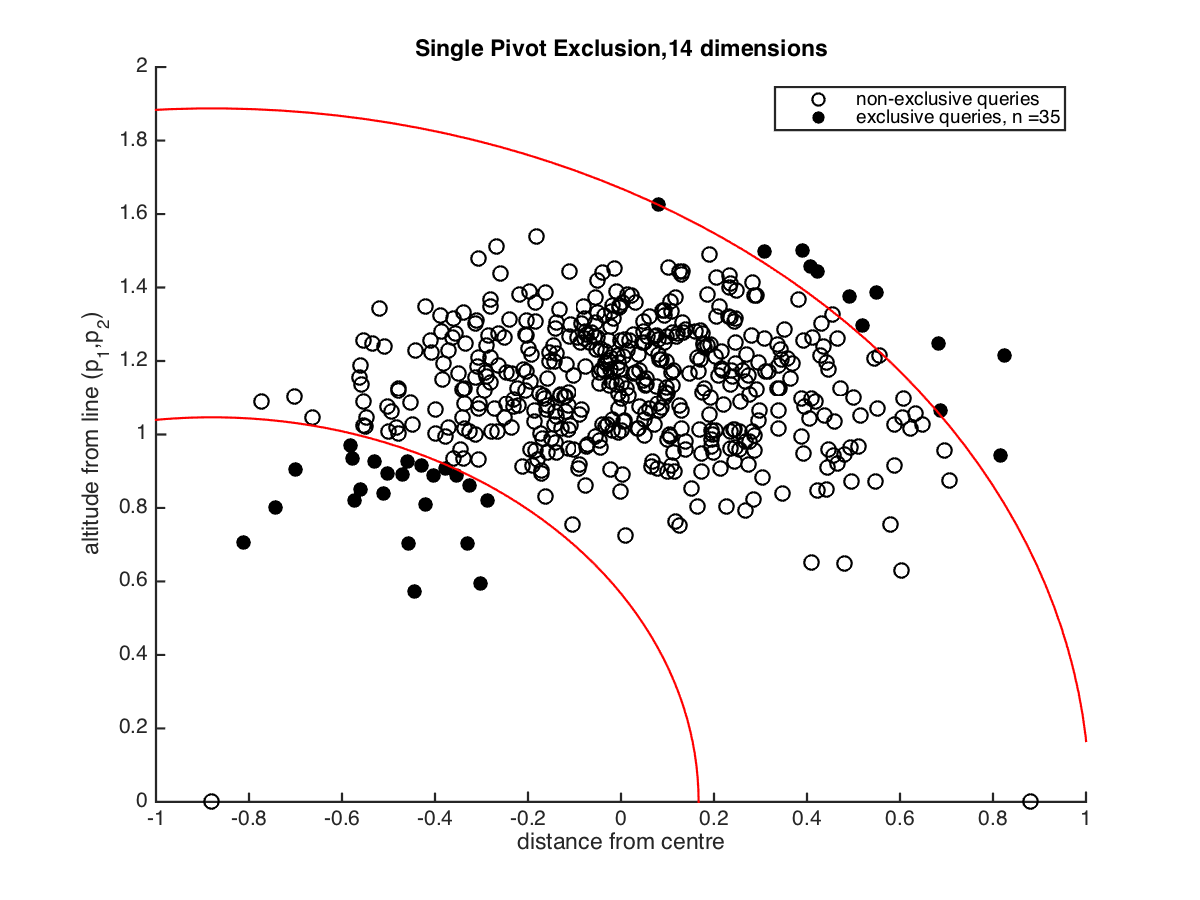}
}

\caption{The effect of dimensionality increase on the three ``power" plots. At 6 dimensions (left hand column)  the   hyperbola can be clearly seen to disadvantage Hyperbolic Exclusion (middle row) against the parallel lines given by Hilbert Exclusion (top row.) At 14 dimensions however, Hyperbolic Exclusion excludes only a handful of  points, and Hilbert Exclusion  achieves significantly  more exclusion than single-point pivoting (bottom row.) Query thresholds are chosen to return one per million objects. }
\label{fig_six_fourteen_dims}
\end{figure}

\section{Conclusions and Further Work}

We have shown that many common metric spaces have a further, stronger, property: namely, as well as the ability to isometrically embed any three points in two-dimensional Euclidean space, they also have the ability to isometrically embed any four points in three-dimensional Euclidean space. We have shown how the stronger geometric guarantee allows more effective metric indexing, and also that any metric space which is isometrically embeddable in Hilbert space has the stronger property. Such spaces include those most commonly used, including spaces of any dimension governed by Euclidean, Jensen-Shannon, Triangular or Cosine distance.

We have shown that, for such spaces, the most popular, state-of-the-art indexing mechanisms have significantly better performance, and that the improvement increases as the dimensionality of the space increases, which is an important result in this field.

However we believe that the so-called four point property will turn out to also be of   value in other areas of similarity search. Although  not yet fully investigated, we have included here the observation that our Hilbert Exclusion has better properties than normal pivot-based exclusion over a single object, and while Hilbert exclusion has the disadvantage of requiring two reference points, it has been seen (for example in monotonous bisector trees) how this extra cost can be amortised by reusing the pivot points. We have also made some early but promising observations that the four-point property can be used to effect beyond indexing structures, for example in the use of locality-sensitive hashing and permutation ordering, which we are currently investigating further.

In essence, almost the entire literature of metric search is based upon the property of  3-embeddability in two dimensional space; almost every derived result in the whole domain can be usefully re-examined in terms of the stronger property of 4-embeddability in three dimensional space.

Finally, it is also the case that any Hilbert space with the four-point property in fact has the ability to embed any $n$ points with $(n-1)$-dimensional Euclidean space; we are currently trying to understand if this property gives rise to further uses within metric indexes.

\label{sec_conclusions}

\bibliographystyle{plain}
\bibliography{bib/connor,bib/cos_projection}

\newpage
\section{Appendices}

\appendix

\section{Algebraic Proof of Weakness}
\label{section_appendix_proof_1}
Here we prove that the Hilbert Exclusion Condition is weaker than the Hyperbolic Exclusion Condition. The intuition behind this is  clear  from the geometric derivation but the algebraic proof is straightforward.

We require to prove that
\[\frac{d(q,p_1)^2 - d(q,p_2)^2}{2\,d(p_1,p_2)} > t\]
is a weaker condition than
\[\frac{d(q,p_1) - d(q,p_2 )}{2} > t\]
for which it is sufficient to show that
\[\frac{d(q,p_1)^2 - d(q,p_2)^2}{2\,d(p_1,p_2)} \ge\frac{d(q,p_1) - d(q,p_2 )}{2}\]
Using the triangle inequality property on $q,p_1$ and $p_2$, this requirement can be stated as
\[\frac{a^2 - b^2}{2\,c} \ge\frac{a- b}{2} \quad,\quad c \le a + b\]
and so
\[\frac{(a+b)(a - b)}{2\,c} \ge \frac{a- b}{2}\]
 which is clear when $c \le a + b$.
%

This proof also neatly demonstrates the fact that the conditions are equivalent only if the query point is colinear with the two pivots $p_1$ and $p_2$; in all other cases, the Hilbert Exclusion Condition is strictly weaker.

\section{IDIMs and Query Thresholds}
\label{section_appendix_1}

\begin{table}[h]
\label{table_idims_thresholds}
\begin{center}
\caption{Intrinsic Dimensionality and Thresholds for Experimental  Spaces}{
{\renewcommand{\arraystretch}{1.16}
\renewcommand{\tabcolsep}{0.15cm}
\begin{tabular}{|l||c|c|c|c|c|c|c|}
\hline
Space	&IDIM	&$t_1$	&$t_2$	&$t_4$	&$t_8$	&$t_{16}$	&$t_{32}$\\
\hline\hline

euc\_6	&7.698	&0.076	&0.085	&0.095	&0.107	&0.120	&0.135\\
\hline
euc\_8	&10.40	&0.149	&0.162	&0.177	&0.193	&0.211	&0.230\\
\hline
euc\_10	&13.36	&0.228	&0.245	&0.262	&0.281	&0.301	&0.323\\
\hline
euc\_12	&16.23	&0.308	&0.327	&0.346	&0.367	&0.388	&0.412\\
\hline
euc\_14	&19.13	&0.386	&0.406	&0.426	&0.448	&0.471	&0.495\\
\hline
jsd\_6	&5.162	&0.022	&0.026	&0.030	&0.035	&0.040	&0.046\\
\hline
jsd\_8	&7.273	&0.045	&0.051	&0.057	&0.064	&0.071	&0.078\\
\hline
jsd\_10	&9.486	&0.067	&0.073	&0.079	&0.086	&0.094	&0.102\\
\hline
jsd\_12	&11.51	&0.084	&0.091	&0.099	&0.107	&0.114	&0.122\\
\hline
jsd\_14	&13.69	&0.103	&0.111	&0.118	&0.126	&0.133	&0.141\\
\hline
tri\_6	&5.754	&0.025	&0.030	&0.035	&0.041	&0.047	&0.055\\
\hline
tri\_8	&8.181	&0.053	&0.060	&0.068	&0.075	&0.083	&0.091\\
\hline
tri\_10	&10.46	&0.078	&0.086	&0.093	&0.101	&0.110	&0.119\\
\hline
tri\_12	&13.02	&0.098	&0.106	&0.116	&0.125	&0.133	&0.142\\
\hline
tri\_14	&15.60	&0.120	&0.129	&0.137	&0.146	&0.155	&0.164\\
\hline

\hline
\end{tabular}}}
\end{center}
\end{table}%

\section{Exclusion Power Results}
\label{section_appendix_exc_pow}

\label{table_power_test_results}
\begin{table}[h]
\begin{center}
\caption{Exclusion Power results for various metrics, spaces and thresholds.}{
{\renewcommand{\arraystretch}{1.3}
\renewcommand{\tabcolsep}{0.15cm}
\begin{tabular}{|c|c||c|c|c||c|c|c||c|c|c|}
\hline
	&&\multicolumn{3}{|c|}{Hyperbolic}	&\multicolumn{3}{|c|}{Hilbert}		&\multicolumn{3}{|c|}{Pivot}						\\
	\hline
	Data Set&IDIM&$t_1$&$t_4$&$t_{16}$&$t_1$&$t_4$&$t_{16}$&$t_1$&$t_4$&$t_{16}$\\
	\hline\hline


euc\_6	&7.64	&59.8	&50.8	&40.7	&80.5	&75.6	&69.4	&74.4	&68.1	&60.4	\\
\hline
euc\_8	&10.5	&31.4	&23.3	&15.8	&62.1	&55.6	&48.3	&51.8	&44.2	&36.0	\\
\hline
euc\_10	&13.3	&12.2	&7.6	&4.3	&44.3	&37.7	&30.8	&31.9	&25.1	&18.7	\\
\hline
euc\_12	&16.1	&3.8	&2.0	&0.9	&29.5	&23.8	&18.4	&17.4	&12.7	&8.6	\\
\hline
euc\_14	&19.0	&0.9	&0.4	&0.2	&18.5	&14.2	&10.3	&8.8	&6.0	&3.8	\\
\hline \hline
jsd\_6	&5.15	&66.1	&54.9	&42.9	&83.8	&77.8	&70.7	&82.4	&75.8	&68.0	\\
\hline
jsd\_8	&7.26	&32.4	&21.7	&13.5	&62.8	&53.9	&45.2	&58.5	&48.8	&39.3	\\
\hline
jsd\_10	&9.39	&11.4	&6.3	&3.0	&42.6	&34.4	&26.4	&36.2	&27.7	&19.8	\\
\hline
jsd\_12	&11.4	&3.5	&1.4	&0.5	&27.4	&19.6	&13.5	&20.8	&13.6	&8.5	\\
\hline
jsd\_14	&13.7	&0.6	&0.2	&0.1	&14.4	&9.5	&6.0	&9.3	&5.4	&3.0	\\
\hline \hline
tri\_6	&5.76	&63.7	&51.9	&39.7	&82.3	&75.8	&68.2	&80.4	&73.1	&64.6	\\
\hline
tri\_8	&8.25	&27.9	&17.6	&10.3	&59.5	&50.1	&41.0	&54.2	&43.9	&34.2	\\
\hline
tri\_10	&10.6	&8.1	&4.1	&1.8	&38.0	&29.7	&21.8	&31.0	&22.8	&15.4	\\
\hline
tri\_12	&13.0	&1.9	&0.6	&0.2	&22.7	&15.3	&9.9	&16.2	&9.8	&5.7	\\
\hline
tri\_14	&15.5	&0.3	&0.1	&0.0	&10.8	&6.6	&3.8	&6.2	&3.3	&1.6	\\
\hline


	\hline
	\end{tabular}}}
\end{center}
\end{table}%

\begin{table}[h]
\label{table_query_results}
\begin{center}
\caption{Indexing Costs for General Hyperplane and Monotonous Hyperplane Tree: mean number of distance calculations per query as percentage of data size ($n = 10^6$).}{
\makebox[\textwidth]{
{\renewcommand{\arraystretch}{1.3}
\renewcommand{\tabcolsep}{0.15cm}
\begin{tabular}{|c||c|c|c||c|c|c||c|c|c||c|c|c|}
\hline
	&\multicolumn{6}{|c||}{Hyperbolic}	&\multicolumn{6}{|c|}{Hilbert}		\\
	\hline
	&\multicolumn{3}{|c||}{GHT}	&\multicolumn{3}{|c||}{MHT}&\multicolumn{3}{|c||}{GHT}	&\multicolumn{3}{|c|}{MHT}\\
	\hline
	Data Set&$t_1$&$t_4$&$t_{16}$&$t_1$&$t_4$&$t_{16}$&$t_1$&$t_4$&$t_{16}$&$t_1$&$t_4$&$t_{16}$\\
	\hline\hline


euc\_6	&0.06	&0.11	&0.20	&0.05	&0.08	&0.15	&0.05	&0.08	&0.14	&0.03	&0.05	&0.10\\
\hline
euc\_8	&0.30	&0.50	&0.84	&0.25	&0.41	&0.68	&0.18	&0.31	&0.55	&0.13	&0.22	&0.40\\
\hline
euc\_10	&1.19	&1.86	&2.91	&1.00	&1.54	&2.33	&0.68	&1.12	&1.87	&0.48	&0.80	&1.35\\
\hline
euc\_12	&3.87	&5.60	&7.97	&3.19	&4.48	&6.25	&2.25	&3.53	&5.48	&1.62	&2.54	&3.97\\
\hline
euc\_14	&9.92	&13.18	&17.26	&7.67	&10.06	&13.17	&6.25	&9.09	&13.02	&4.47	&6.57	&9.53\\
\hline
\hline
tri\_6	&0.05	&0.11	&0.21	&0.04	&0.08	&0.16	&0.04	&0.07	&0.15	&0.02	&0.05	&0.11\\
\hline
tri\_8	&0.40	&0.78	&1.41	&0.32	&0.62	&1.10	&0.23	&0.48	&0.92	&0.17	&0.35	&0.69\\
\hline
tri\_10	&1.95	&3.29	&5.37	&1.66	&2.73	&4.36	&1.11	&2.05	&3.71	&0.84	&1.57	&2.87\\
\hline
tri\_12	&6.10	&9.84	&14.49	&5.25	&8.24	&12.04	&3.74	&6.86	&11.27	&2.92	&5.43	&9.04\\
\hline
tri\_14	&16.63	&23.11	&30.57	&13.95	&19.45	&26.06	&12.02	&18.57	&26.52	&9.68	&15.24	&22.25\\
\hline
\hline
jsd\_6	&0.05	&0.10	&0.20	&0.04	&0.08	&0.15	&0.04	&0.07	&0.15	&0.02	&0.05	&0.11\\
\hline
jsd\_8	&0.32	&0.63	&1.15	&0.26	&0.51	&0.92	&0.20	&0.40	&0.78	&0.14	&0.29	&0.58\\
\hline
jsd\_10	&1.50	&2.58	&4.29	&1.35	&2.22	&3.61	&0.90	&1.64	&2.99	&0.68	&1.25	&2.31\\
\hline
jsd\_12	&4.67	&7.68	&11.47	&4.17	&6.62	&9.76	&2.84	&5.27	&8.71	&2.22	&4.15	&6.97\\
\hline
jsd\_14	&12.4	&17.67	&23.9	&10.77	&15.17	&20.57	&8.62	&13.62	&19.97	&6.94	&11.13	&16.69\\
\hline

	\hline
	\end{tabular}}}}
\end{center}

\end{table}%

\newpage
\end{document}